\documentclass[a4paper,UKenglish,cleveref, autoref, thm-restate,]{lipics-v2021-arxiv}
\pdfoutput=1 %
\hideLIPIcs  %

\usepackage{hyperref}
\hypersetup{hidelinks}

\usepackage{mathtools, amsthm, amsfonts, amssymb}
\usepackage{booktabs}

\nolinenumbers

\theoremstyle{plain}

\newtheorem*{theorem*}{Theorem}
\newtheorem*{proposition*}{Proposition}

\theoremstyle{definition}
\newtheorem{problem}[theorem]{Problem}

\newtheorem*{definition*}{Definition}

\newcommand{\bigO}{\mathcal{O}}

\newcommand{\smallo}{o} %

\newcommand{\eval}[1]{\mathsf{Eval}_{#1}}

\newcommand*{\cP}{\mathcal{P}}

\newcommand*{\cX}{\mathcal{X}}

\newcommand*{\Z}{\mathbb{Z}}

\newcommand*{\FF}{\mathbb{F}}
\newcommand*{\NN}{\mathbb{N}}

\DeclareMathOperator{\rank}{rank}

\newcommand*{\poly}{\mathrm{poly}}

\newcommand*{\supp}{\mathrm{supp}}

\newcommand{\dcc}[1]{\mathrm{D}_{#1}}

\newcommand*{\D}{\mathrm{D}}

\newcommand{\ra}{\rightarrow}
\newcommand{\cck}[2]{c_{#1,\angle}({#2})}
\newcommand{\cc}[1]{c_{\angle}({#1})}
\newcommand{\rc}[1]{r_{\!\angle}({#1})}
\newcommand{\rck}[2]{r_{#1,\angle}({#2})}

\newcommand{\abs}[1]{\left|{#1}\right|} %

\DeclarePairedDelimiterX{\inner}[2]{\langle}{\rangle}{#1, #2}
\DeclareMathOperator{\subrank}{Q}
\DeclareMathOperator{\trank}{R}
\DeclareMathOperator{\asympsubrank}{\tilde{Q}}
\DeclareMathOperator{\slicerank}{SR}
\DeclareMathOperator{\asyslicerank}{\tilde{SR}}

\newcommand{\cor}{\mathrm{cor}}
\newcommand{\capset}{\mathrm{cap}}
\newcommand{\IM}{\mathrm{IM}}

\title{Larger Corner-Free Sets from Combinatorial Degenerations}

\author{Matthias Christandl}{Department of Mathematical Sciences, University of Copenhagen}{christandl@math.ku.dk}{}{European Research Council (ERC Grant Agreement No.~818761), the VILLUM FONDEN via the QMATH Centre of Excellence (Grant No.~10059) and the QuantERA ERA-NET Cofund in Quantum Technologies implemented within the European Union’s Horizon 2020 Programme (QuantAlgo project) via the Innovation Fund Denmark}

\author{Omar Fawzi}{Univ.~Lyon, ENS Lyon, UCBL, CNRS, Inria, LIP}{omar.fawzi@ens-lyon.fr}{}{European Research Council (ERC Grant Agreement No. 851716)}

\author{Hoang Ta}{Univ.~Lyon, ENS Lyon, UCBL, CNRS, Inria, LIP}{duy-hoang.ta@ens-lyon.fr}{}{European Research Council (ERC Grant Agreement No.~851716), LABEX MILYON (ANR-10-LABX-0070) of Université de Lyon, within the program ``Investissements d’Avenir'' (ANR-11-IDEX-0007) operated by the French National Research Agency (ANR)}

\author{Jeroen Zuiddam}{Korteweg-de Vries Institute for Mathematics, University of Amsterdam}{j.zuiddam@uva.nl}{}{}

\authorrunning{M.~Christandl, O.~Fawzi, H.~Ta and J.~Zuiddam} %

\Copyright{Matthias Christandl, Omar Fawzi, Hoang Ta and Jeroen Zuiddam} %



\category{} %

\relatedversion{This is the full version (including all proofs) of the paper with the same title that will  appear in the proceedings of the conference Innovations in Theoretical Computer Science (ITCS)~2022. The bounds on corner-free sets in this paper improve earlier results in the preprint~\cite{christandl2021communication} by the same authors.
} %

\acknowledgements{The authors thank Wenjie Fang, Stéphan Thomassé, Adi Shraibman and Pascal Koiran for helpful discussions.}

\EventEditors{John Q. Open and Joan R. Access}
\EventNoEds{2}
\EventLongTitle{42nd Conference on Very Important Topics (CVIT 2016)}
\EventShortTitle{CVIT 2016}
\EventAcronym{CVIT}
\EventYear{2016}
\EventDate{December 24--27, 2016}
\EventLocation{Little Whinging, United Kingdom}
\EventLogo{}
\SeriesVolume{42}
\ArticleNo{23}

\begin{document}

\maketitle

\begin{abstract}
    There is a large and important collection of Ramsey-type combinatorial problems, closely related to central problems in complexity theory, that can be formulated in terms of the asymptotic growth of the size of the maximum independent sets in powers of a fixed small (directed or undirected) hypergraph, also called the Shannon capacity. An important instance of this is the corner problem studied in the context of multiparty communication complexity in the Number On the Forehead (NOF) model (and other important instances are the cap set problem in additive combinatorics and the USP capacity problem in the complexity theory of matrix multiplication). Versions of this problem and the NOF connection have seen much interest (and progress) in recent works of Linial, Pitassi and Shraibman (ITCS~2019) and Linial and Shraibman (CCC~2021).

    We introduce and study a general algebraic method for lower bounding the Shannon capacity of directed hypergraphs via \emph{combinatorial degenerations}, a combinatorial kind of ``approximation'' of subgraphs that originates from the study of matrix multiplication in algebraic complexity theory (and which play an important role there) but which we use in a novel way. 
    
    Using the combinatorial degeneration method, we make progress on the corner problem by explicitly constructing a corner-free subset in $\FF_2^n \times \FF_2^n$ of size $\Omega(3.39^n/\poly(n))$, which improves the previous lower bound~$\Omega(2.82^n)$ of Linial, Pitassi and Shraibman (ITCS~2019) and which gets us closer to the best upper bound $4^{n - o(n)}$. Our new construction of corner-free sets implies an improved NOF protocol for the \emph{Eval problem}. In the Eval problem over a group $G$, three players need to determine whether their inputs $x_1, x_2, x_3 \in G$ sum to zero. We find that the NOF communication complexity of the Eval problem over $\FF_2^n$ is at most $0.24n + \bigO(\log n)$, which improves the previous upper bound~$0.5n + \bigO(\log n)$.
    
    Finally, we investigate the existing tensor methods for upper bounding the Shannon capacity (including slice rank, subrank, analytic rank, geometric rank, and G-stable rank). We find that these methods have strong limitations caused by the existence of large induced matchings. In particular, this implies a strong barrier for these methods to prove nontrivial upper bounds for the corner problem over any group $G$ (and in particular for $G = \FF_2^n$ to get an upper bound below $4^{n - o(n)}$).
\end{abstract}

\section{Introduction}
This paper is about constructing special combinatorial objects, namely ``corner-free sets'' in~$\FF_p^n \times \FF_p^n$, motivated (besides their inherent interest) by central problems in communication complexity, specifically in the study of the number on the forehead (NOF) model of communication introduced by Chandra, Furst and Lipton~\cite{CFL83}. 
There has been much interest in (and progress on) the corner problem, variations of the problem, and connections to NOF communication, in particular in the recent works of
Shraibman~\cite{SHRAIBMAN201844},
Linial, Pitassi and Shraibman~\cite{LPS18}, Viola~\cite{DBLP:journals/sigact/Viola19}, %
Alon and Shraibman~\cite{alon2020algorithmic}, and
Linial and Shraibman~\cite{DBLP:conf/coco/LinialS21,linial2021larger}. In the recent work of Linial and Shraibman~\cite{DBLP:conf/coco/LinialS21} a construction of large corner-free sets in $[N] \times [N]$ was obtained in an elegant manner by designing efficient NOF communication protocols for a specific communication problem (much like the upcoming Eval problem). %
We take a different, algebraic approach to the corner problem, and make progress on the corner problem over $\FF_p^n \times \FF_p^n$ by introducing in this area a new algebraic method via \emph{combinatorial degeneration}.

\subsubsection*{NOF communication complexity}

The NOF model is very rich in terms of connections to Ramsey theory and additive combinatorics~\cite{beigel2006multiparty,SHRAIBMAN201844,LPS18,DBLP:conf/coco/LinialS21,linial2021larger}, as well as applications to boolean models of compution such as branching programs and boolean circuits~\cite{CFL83,Bei94}. 
The goal in the NOF model is for $k$ players to compute a fixed given function $F : \cX_1 \times \cdots \times \cX_k \to \{0,1\}$ on inputs $(x_1, \dots, x_k) \in \cX_1 \times \cdots \times \cX_k$ where player $i$ has access to input $x_j$ for all $j \neq i$ but no access to input~$x_i$. 
For $k = 2$, this model coincides with the standard two-party communication model of Yao~\cite{Yao79}, but when $k \geq 3$, the shared information between the players makes this model surprisingly powerful~\cite{Gro94,BGKL04,ACFN15,CS14}, %
and fundamental problems remain open. 
For instance, a sufficiently strong lower bound for an explicit function~$F$ for $k \geq \mathrm{polylog}(n)$ players with $n = \log |\cX_i|$ would imply a breakthrough result in complexity theory, namely a lower bound on the complexity class $\textsf{ACC}^0$.

\subsubsection*{NOF complexity of the Eval problem}
A central open problem in the theory of NOF communication is to construct an explicit function for which randomized protocols are significantly more efficient than deterministic ones~\cite{beame2007separating}. 
A well-studied candidate for this separation (for $k=3$) is the function~$\eval{\FF_2^n}$, which is defined by $\eval{\FF_2^n}(x_1, x_2, x_3) = 1$ if and only if $x_1 + x_2 + x_3 = 0$, where the additions are all in~$\FF_2^n$. Thus the Eval problem naturally generalizes the equality problem for $k=2$.
It is known that in the \emph{randomized} setting, the standard protocol for the two-party equality problem that uses $\bigO(1)$ bits of communication works in the same way for three parties for the Eval problem. 
However, in the \emph{deterministic} setting, the communication complexity $\dcc{3}(\eval{\FF_2^n})$ remains wide open: the best known lower bound $\Omega(\log \log n)$ follows from the work of Lacey and McClain~\cite{Lacey_07} and, before this work, the best upper bound was $0.5n + \bigO(\log n)$~\cite{ACFN15}. 
\subsubsection*{Corner problem in combinatorics, and connection to the Eval problem}
Chandra, Furst and Lipton~\cite{CFL83} found that the deterministic communication complexity of many problems in the NOF model can be recast as Ramsey theory problems. 
In particular, and this leads to the problem of interest in this paper, the (deterministic) communication complexity of~$\eval{\FF_2^n}$ can be characterized in terms of corner-free subsets of $\FF_2^{n} \times \FF_2^{n}$, as follows. 
We call any triple of elements $(x,y), (x + \lambda, y), (x, \lambda + y)$ for $x, y, \lambda \in \FF_2^n$ a \emph{corner}. A subset $S \subseteq \FF_2^{n} \times \FF_2^{n}$ is called \emph{corner-free} if it does not contain any nontrivial corners (where nontrivial means that~$\lambda \neq 0$).
Denoting by $\rc{\FF_2^n}$ the size of the largest corner-free set in $\FF_2^n \times \FF_2^n$, the communication complexity of $\eval{\FF_2^n}$ equals $\log(4^n/\rc{\FF_2^n})$ up to a $\bigO(\log n)$ additive term, which provides the close connection between the Eval problem in NOF communication and the corner problem in combinatorics. In particular, large corner-free sets in $\FF_2^n \times \FF_2^n$ correspond to efficient protocols for $\eval{\FF_2^n}$.

\subsubsection*{General paradigm: Shannon capacity of hypergraphs}
The point of view we will take (and the general setting in which the methods we introduce will apply) is to regard the corner problem as a  Shannon capacity problem of directed hypergraphs.
Namely, the size $\rc{\FF_2^n}$ of the largest corner-free set in $\FF_2^n \times \FF_2^n$ can be characterized as the independence number of a (naturally defined) directed $3$-uniform hypergraph with $4^n$ vertices.\footnote{As usual an \emph{independent set} of a hypergraph is a subset $S$ of vertices such that no hyperedge has all its vertices in~$S$.} 
This hypergraph has a recursive form: it is obtained by taking the $n$-th power of a fixed (directed) hypergraph $H_{\cor,\FF_{2}}$ on $4$ vertices. (We discuss this in more detail in Section~\ref{application}.) 
The asymptotic growth of $\rc{\FF_2^n}$ as $n \to \infty$ is characterized by the Shannon capacity $\Theta(H_{\cor,\FF_{2}})$ of the corner hypergraph~$H_{\cor,\FF_{2}}$.\footnote{In the setting of directed graphs, also the term \emph{Sperner capacity} (typically applied to the complement graph)~\cite{GARGANO90,Gargano1993} is used for what we call Shannon capacity.}
That is, we have $\rc{\FF_2^n} = \Theta(H_{\cor,\FF_{2}})^{n - o(1)}$.
In this way, proving the strict upper bound $\Theta(H_{\cor,\FF_{2}}) < 4$ is equivalent to proving a linear lower bound on the communication complexity of $\eval{\FF_2^n}$. 
Many other Ramsey type problems can be expressed as the Shannon capacity of some fixed hypergraph, such as the Cap Set problem that saw a recent breakthrough by Ellenberg and Gijswijt~\cite{EG17} following Croot, Lev and Pach~\cite{croot2017progression}, and the Uniquely Solvable Puzzle (USP) problems that were put forward in the ``group-theoretic approach'' to the matrix multiplication problem~\cite{1530730,Alon2013}.

\subsection{Is the complexity of the Eval problem maximal?}
Let us discuss the open problem that motivates our work, and that is central in NOF communication complexity and combinatorics (throught the aforementioned connections). This problem asks whether or not the complexity of the Eval probem is ``maximal'', or in other words, whether or not there are corner-free sets in~$\FF_2^n \times \FF_2^n$ that have ``sub-maximal'' size:
\begin{problem}\label{eq:main_question}
Are the following three statements (which we know are equivalent\footnote{The equivalence among the three formulations is standard and follows from Lemma~\ref{corner_nof_eval}, Proposition~\ref{prop:coloring_indep_set} and Lemma~\ref{lem:corner_indep_set} further on in the paper. We will mainly use the formulation in terms of Shannon capacity (see Definition~\ref{Shannoncapacity} for a precise definition).}) true?
\begin{itemize}
	\item $\dcc{3}(\eval{\FF_2^n}) = \Omega(n)$
	\item $\rc{\FF_2^n} \leq \bigO(c^n) \text{ for some } c < 4$
	\item $\Theta(H_{\cor,\FF_{2}}) < 4$.
\end{itemize}

\end{problem}
Here the best capacity lower bound before our work was $\Theta(H_{\cor, \FF_2}) \geq \sqrt{8}$ by Linial, Pitassi and Shraibman~\cite[Cor.~24 in the ITCS version]{LPS18}, obtained by explicit construction of an independent set in the second power of the relevant hypergraph, which in turn leads to the bounds $\D_3(\eval{\FF_2^n}) \leq 0.5n + \bigO(\log n)$ and $\rc{\FF_2^n} \geq \sqrt{8}^n$.

In the above we may naturally generalize $\FF_2^n$ to $\FF_p^n$ or even to $G^n$, where $G$ is an arbitrary abelian group, so that Problem~\ref{eq:main_question} is a special case of the more general problem:
\begin{problem}\label{prob:gen}
	Are the following three statements (which we know are equivalent) true?
\begin{itemize}
	\item $\dcc{3}(\eval{G^n}) = \Omega(n)$ 
	\item $\rc{G^n} \leq \bigO(c^n) \text{ for some } c < |G|^2$
	\item $\Theta(H_{\cor,G}) < |G|^2$.
\end{itemize}
\end{problem}
Our goal in this paper, motivated by the connections as remarked earlier, is to make progress on above problems via new algebraic methods.

\subsection{Lower bounds for the corner problem (and other problems) from combinatorial degeneration}
Our main result is progress on Problem~\ref{prob:gen} by proving new lower bounds for the corner problem over the groups $\FF_2$ and $\FF_3$, which we arrive at via a new method to lower bound the Shannon capacity of directed hypergraphs. Equivalently, in the language of communication complexity, we obtain improved protocols for the Eval problem.

The lower bound of Linial, Pitassi and Shraibman \cite{LPS18} for the corner problem was obtained by explicit construction of an independent set (i.e.~a set that does not contain edges) in the second power of a hypergraph, which is the natural approach for such lower bounds. We improve on this bound by observing that it is actually sufficient to construct a set which does not contain ``cycles''. 
For graphs, the notion of cycle is clear but for hypergraphs there are many possible definitions, and we initiate a careful study of this (and believe that this will be a worthwile avenue for further study independently).
Here, to get new bounds we use the notion of \emph{combinatorial degeneration} %
to model such a ``cycle''. We will say more about this in a moment.

Using the combinatorial degeneration method on the corner hypergraphs that characterize the corner problem %
we find new bounds for Problem~\ref{prob:gen} for the groups $\FF_2^n$ and $\FF_3^n$. %
These are as follows (in the three equivalent forms):
\begin{theorem}[Thm.~\ref{cor:lower_sqrt10}]\label{th:F2-intro}
For the corner and Eval problem over $\FF^n_2$ we have:
\begin{itemize}%
	\item $\dcc{3}(\eval{\FF_2^n}) \leq 0.24n + \bigO(\log n)$
	\item $\rc{\FF_2^n} \geq \frac{\sqrt[3]{39}^{n}}{\poly(n)}$ 
	\item $\Theta(H_{\cor,\FF_{2}}) \geq \sqrt[3]{39}$
\end{itemize}
\end{theorem}

\begin{theorem}[Thm.~\ref{cor:lower_F3}]\label{th:F3-intro}
For the corner and Eval problem over $\FF^n_3$ we have:
\begin{itemize}%
	\item $\dcc{3}(\eval{\FF_3^n}) \leq 0.37n + \bigO(\log n)$
	\item $\rc{\FF_3^n} \geq \frac{7^{n}}{\poly(n)}$
	\item $\Theta(H_{\cor,\FF_{3}}) \geq 7$.
\end{itemize}
\end{theorem}

Let us discuss on a high level the history and ideas behind the combinatorial degeneration method. 
Combinatorial degeneration is an existing concept from algebraic complexity theory. It was (in a slightly different form) introduced and studied by Strassen in \cite[Section 6]{Strassen1987RelativeBC}.\footnote{\emph{Degeneration} of tensors is a powerful approximation notion in the theory of tensors. \emph{Combinatorial degeneration} is the ``combinatorial'' or ``torus'' version of this kind of approximation. Combinatorial degeneration was introduced by Bürgisser, Clausen and Shokrollahi~\cite[Definition~15.29]{PMM97} based on the notion of M-degeneration for tensors defined and studied by Strassen in \cite{Strassen1987RelativeBC}.} (For the formal definition of combinatorial degeneration, see Definition~\ref{def:comb-degen}.)
Strassen's original application of combinatorial degeneration was to study matrix multiplication, namely to prove the fundamental result that surprisingly many independent scalar multiplications can be reduced (in an appropriate algebraic manner) to matrix multiplication \cite[Theorem~6.6]{Strassen1987RelativeBC}.\footnote{Strassen's result is asymptotically optimal. Strassen's proof resembles Behrend's construction of arithmetic-progression-free sets. Also note that this is precisely the opposite of the problem of reducing matrix multiplication to as few independent scalar multiplications as possible. The latter corresponds precisely to the arithmetic complexity of matrix multiplication.} 
Strassen then used this result  to prove his Laser method \cite[Section 7]{Strassen1987RelativeBC}, vastly generalizating the method that Coppersmith and Winograd had introduced in their construction of matrix multiplication algorithms \cite{coppersmith1987matrix}.\footnote{The book~\cite[Definition~15.29 and Lemma~15.31]{PMM97} gives a different proof of the Laser method which relies even more strongly on combinatorial degeneration.} %

Combinatorial degeneration was used more broadly to construct large induced matchings in the setting of important combinatorial problems, namely the Sunflower problem by Alon, Shpilka and Umans~\cite[Lemma~3.9]{Alon2013} and the Cap Set problem by Kleinberg, Sawin and Speyer~\cite{Sawin}. These results are often referred to as the ``multicolored'' versions of the problem at hand, as opposed to the ``single color'' version.
These ideas were developed further in the context of matrix multiplication barriers by Alman and Williams~\cite[Lemma~6]{DBLP:conf/innovations/AlmanW18} and in the study of tensors by Christandl, Vrana and Zuiddam~\cite[Theorem~4.11]{CVZ18}.

Crucially, all of the above applications use combinatorial degeneration to construct \emph{induced matchings} in ($k$-uniform $k$-partite) hypergraphs.
However, we use combinatorial degeneration in a novel manner to construct \emph{independent sets} in hypergraphs instead of induced matchings. In this context an independent set should be thought of as a symmetric induced matching. Constructing large independent sets is a much harder task than constructing large induced matchings, as witnessed by the fact that the ``multicolored'' cap set problem is solved~\cite{Sawin} while its ``single color'' version is not. Similarly, for the corner problem, as we will discuss in Section~\ref{sec:intro_barrier}, the asymptotic growth of the largest induced matching can be shown to be maximal, whereas the main question of study of this paper is whether the same holds for the largest independent set.
We expect our new way of using combinatorial degeneration to be useful in the study of other problems besides the corner problem as well, and thus think it is of independent interest. 

On a more technical level, combinatorial degeneration is a notion that compares sets of $k$-tuples by means of algebraic conditions. Our ``universe'' is $I_1 \times \cdots \times I_k$ where $I_1, \ldots, I_k$ are finite sets.
Then for sets $\Phi \subseteq \Psi \subseteq I_1 \times \dots \times I_k$ we say that $\Phi$ is a \emph{combinatorial degeneration} of $\Psi$, and write~$\Psi \unrhd \Phi$, if there are maps $u_i:I_i \ra \Z$ such that for every $x = (x_1, \ldots, x_k) \in I_1 \times \dots \times I_k$, if~$x \in \Psi \setminus \Phi$, then $\sum_{i=1}^{k}u_i(x_i)>0$, and if $x \in \Phi$, then $\sum_{i=1}^{k}u_i(x_i) = 0$. Thus the maps $u_i$ together are able to distinguish between the elements in the set $\Phi$ (which may be thought of as our ``goal'' set, i.e. a set with good properties) and the elements in the difference $\Psi \setminus \Phi$.
As a quick example of a combinatorial degeneration, let
\begin{align*}
\Phi &= \{(0,0,0),(1,1,0),(1,0,1) \}, \\
\Psi &= \{(0,0,0),(1,1,0),(1,0,1),(0,1,1)\}.
\end{align*}
Then we find a combinatorial degeneration $\Psi \unrhd \Phi$ by defining the maps $u_i : \{0,1\} \to \Z$ simply by setting $u_1(0) = u_2(0) = u_3(0) = 0$, and $u_1(1) = -1$, $u_2(1) = u_3(1) = 1$. 

We apply the idea of combinatorial degeneration in the following fashion to get Shannon capacity lower bounds:

\begin{theorem}[Combinatorial degeneration method, Theorem~\ref{thm:combinatorial_degeneration}]\label{thm:combinatorial_degeneration-intro}
Let $H = (V,E)$ be a directed $k$-uniform hypergraph. Let $S \subseteq V$ be a subset of vertices. Define the sets 
\[
	\Psi = E \cup\{(v,\ldots, v) : v\in V\}
\]
and
\[
	\Phi = \{(v,\ldots, v) : v \in S\}.
\]
Suppose that $\Psi \unrhd \Phi$ is a combinatorial degeneration. Then we get the Shannon capacity lower bound $\Theta(H) \geq \abs{S}$.
\end{theorem}
In other words, whereas $S$ in the statement of Theorem~\ref{thm:combinatorial_degeneration-intro} may not be an independent set, we can via the algebraic conditions of combinatorial degeneration construct an independent set in the $n$th power of the hypergraph of size approaching $|S^{\times n}|$. Namely, the algebraic conditions allow us to select such an independent set using a natural type analysis of the labels given by the maps $u_i$. Thus we may think of a set $S$ as above as an \emph{approximative} independent set, which asymptotically we can turn into an actual independent set by means of Theorem~\ref{thm:combinatorial_degeneration-intro}.

We note that, whereas it is relatively simple to verify for a given set $S$ that $\Psi \unrhd \Phi$ holds (with the notation of Theorem~\ref{thm:combinatorial_degeneration-intro}) via linear programming, it is seems much harder to \emph{find} a large set $S \subseteq V$ for which $\Psi \unrhd \Phi$, given $H$. We obtain our best lower bounds via an integer linear programming approach. The resulting combinatorial degenerations that we find are explicit and checkable by hand.

We have yet to develop structural understanding of how the above combinatorial degenerations that exhibit the new capacity lower bounds arise (and we feel that deeper understanding of this may lead to more progress or even solve the corner problem), and leave the investigation of further generalizations and improvements to future work. 
As a partial remedy to our limited understanding, we introduce the \emph{acyclic method} as a tool to construct combinatorial degenerations. While the acyclic method does not recover the bounds of Theorem~\ref{th:F2-intro} and Theorem~\ref{th:F3-intro}, it has the merits of being transparent and simple to apply. The acyclic method involves another notion of a set wihtout ``cycles'', which implies a combinatorial degeneration, but whose conditions are simpler to check.

\subsection{Lower bounds for the corner problem from the probabilistic method}
We employ the probabilistic method to find the following very general bound for the corner problem over arbitrary abelian groups.
\begin{theorem}[Prop.~\ref{prop:simple_lowerbound}]\label{th:prob-intro}
For the corner and Eval problem over an arbitrary abelian group~$G$ we have
\begin{itemize}
	\item $\dcc{3}(\eval{G^n}) \leq \frac{\log |G|}{2} n + \bigO(\log n)$
	\item $\rc{G^n} \geq \frac{|G|^{3n/2}}{\poly(n)}$
	\item $\Theta(H_{\cor,G}) \geq |G|^{3/2}$.
\end{itemize}
\end{theorem}
This general bound applied to the special cases $G = \FF_2$ and $G = \FF_3$ does not quite match the bounds in Theorem~\ref{th:F2-intro} and Theorem~\ref{th:F3-intro}, respectively. However, applied to the special case $G = \FF_2$ we do recover the lower bound~$\sqrt{8}$ of~\cite[Cor.~24 in the ITCS version]{LPS18}.

Using the same techniques we gain insight about the high-dimensional version of the corner problem and Eval problem and what happens when the number of players grows. For an arbitrary abelian group $G$, a $k$-dimensional corner over $G$ is naturally defined as a set of $(k+1)$ points in $(G)^{\times k}$ of the form 
\[
\{(x_1, x_2, \ldots, x_k), (\lambda + x_1, x_2, \ldots, x_k), \ldots, (x_1, x_2, \ldots, \lambda + x_k)\}
\]
where $x_i, \lambda \in G^n$. A subset $S \subseteq G^{\times k}$ is called corner-free if it does not contain any nontrivial corners (where nontrivial again means $\lambda \neq 0$). We denote by $\rck{k}{G}$ the size of the largest ($k$-dimensional) corner-free set. Just like the $k=2$ case corner-free sets correspond to independent sets in a naturally defined $(k+1)$-uniform directed hypergraph $H_{k, \cor, G}$. With the probabilistic method (extending Theorem~\ref{th:prob-intro}), we find that when the $k$ goes to infinity, the capacity of $H_{k, \cor, G}$ becomes essentially maximal. As a consequence if $k$ grows with $n$ (e.g., $k = \log n$) we find that the NOF complexity of the corresponding $k$-player Eval problem becomes \emph{sub-linear}.

\begin{theorem}[Rem.~\ref{rem:multiplayer}]\label{th:prob-higher-intro} Let $G$ be a finite abelian group. Then
	\begin{itemize}
		\item $\Theta(H_{k, \cor, G}) \geq |G|^{k - \tfrac1k}$
		\item $\D_{k+1}(\eval{G^n}) \leq \frac{n}{k} \log |G| + \bigO(\log n) + k$
	\end{itemize}
\end{theorem}
Thus we learn that to prove a linear lower bound on $\D_{k}(\eval{G^n})$ for any given $G$ (say for~$G = \FF_p$) it is important to keep $k$ constant.

\subsection{Limitations of tensor methods for proving upper bounds for the corner problem}
\label{sec:intro_barrier}
Our second result is a strong limitation of current tensor methods to effectively upper bound the Shannon capacity of hypergraphs. This limitation is caused by induced matchings and applies to various combinatorial problems including the corner problem. We use a method of Strassen to show that these limitations are indeed very strong for the corner problem.

In order to elaborate on these results let us first give an overview of upper bound methods.
The general question of upper bounds on the Shannon capacity of hypergraphs is particularly well-studied in the special setting of undirected graphs, from which the name ``Shannon capacity'' comes: it in fact corresponds to the zero-error capacity of a channel~\cite{sha56}. 
Even for undirected graphs, it is not clear how to compute the Shannon capacity in general, but some methods were developed to give upper bounds. 
The difficulty is to find a good upper bound on the largest independent set that behaves well under the product $\boxtimes$. 
For undirected graphs, the best known methods are the Lov\'asz theta function~\cite{Lo79} and the Haemers bound which is based on the matrix rank~\cite{HW79}. For hypergraphs, we only know of algebraic methods that are based on various notions of tensor rank, and in particular the slice rank~\cite{TS16} (which was used and studied extensively in combinatorics, in the context of cap sets \cite{Tao16,Sawin}, 
sunflowers \cite{naslund2017upper} and right-corners \cite{Naslund}), and similar notions like the analytic rank \cite{gowers2011linear, Lovett, Briet2019SubspacesOT}, the geometric rank \cite{kopparty_et_al}, and the G-stable rank \cite{derksen2020gstable}. Even though the slice rank is not multiplicative under $\boxtimes$ it is possible to give good upper bounds on the asymptotic slice rank via an asymptotic analysis \cite{TS16}, which is closely related to the Strassen support functionals~\cite{Str91} or the more recent quantum functionals~\cite{CVZ18}.

Most of the rank-based bounds actually give upper bounds on the size of induced matchings and not only on the size of independent sets. It is simple and instructive to see this argument in the setting of undirected graphs. 
For a given graph $H = (V,E)$, let $A$ be the adjacency matrix in which we set all the diagonal coefficients to $1$. 
Then for any independent set $I \subseteq V$, the submatrix $(A_{i,j})_{i,j \in I}$ of $A$ is the identity matrix and as a result $|I| \leq \rank(A)$. 
As the matrix rank is multiplicative under tensor product, we get $\Theta(H) \leq \rank(A)$.
Observe that this argument works equally well if we consider an induced matching instead of an independent set.
An induced matching of size $s$ of the graph $H = (V,E)$ can be defined by two lists of vertices  $I_1(1), \dots, I_1(s)$ and $I_2(1), \dots, I_2(s)$ of size $s$ such that for any $\alpha, \beta \in \{1, \dots, s\}$ we have
\begin{align*}
((I_1(\alpha), I_2(\beta)) \in E \text{ or } I_1(\alpha) = I_2(\beta)) \quad \Longleftrightarrow \quad \alpha = \beta \ . 
\end{align*}
In other words, the submatrix $(A_{i,j})_{i \in I_1, j \in I_2}$ is an identity matrix, which also implies that $s \leq \rank(A)$. As such, the matrix rank is an upper bound on the asymptotic maximum induced matching. 
Tensor rank methods such as the subrank, slice rank, analytic rank, geometric rank and G-stable rank also provide upper bounds on the asymptotic maximum induced matching.

Using a result of Strassen~\cite{Str91}, we show that there is an induced matching of the $n$-th power of $H_{\cor,\FF_{2}}$ of size $4^{n - o(1)}$. This establishes a \emph{barrier} on many existing tensor methods (such as slice rank, subrank, analytic rank,~etc.) to make progress on Problem~\ref{eq:main_question}. In fact, this result holds more generally for any abelian group $G$:

\begin{theorem}[Cor.~\ref{cor:cor-tight}]
For any abelian group $G$, the hypergraph $H_{\cor, G}^{\boxtimes n}$ has an induced matching of size $|G|^{2n - o(n)}$.
In other words, for any $n \geq 1$, there exist lists $I_1, I_2, I_3 \subseteq G^n \times G^n$ of size $s(n) = |G|^{2n - o(n)}$ such that the following holds. For any $\alpha, \beta, \gamma \in \{1, \dots, s(n)\}$
\begin{align}
\left(I_1(\alpha), I_2(\beta), I_3(\gamma) \right) \text{ forms a corner} \quad \Longleftrightarrow \quad \alpha = \beta = \gamma \ .
\end{align}
\end{theorem}
We prove this result by establishing in Theorem~\ref{prop:corner-tight} that the adjacency tensor of the hypergraph $H_{\cor, G}$ is \emph{tight} (see Definition~\ref{Tighttensor}). Strassen showed in~\cite{Str91} that for tight sets, the asymptotic induced matching number is characterized by the support functionals. By computing the support functionals for the relevant tensors, we establish the claimed result in Corollary~\ref{cor:cor-tight}. Note that if we could ensure that $I_1 = I_2 = I_3$, this would solve Problem~\ref{eq:main_question}. We computed the maximum independent set and maximum induced matching for $H_{\cor, \FF_2}^{\boxtimes n}$ for small powers $n = 1,2,3$ (see Table~\ref{tab:max_induced_matching}) and we found that the maximum independent set is strictly smaller than the maximum induced matching for $n = 2$ and $n = 3$. This motivates the search for methods that go beyond the maximum induced matching barrier. For comparison, we also give the analogous numbers for the cap set hypergraph~$H_{\capset}$ (which is an undirected hypergraph), where, interestingly, the maximum independent set and the maximum induced matching are equal. 

\begin{table}
	\begin{center}
	\begin{tabular}{ ccc } 
		\multicolumn{3}{c}{$H_{\capset}^{\boxtimes n}$} \\[0.5em]
		\toprule
		$n$ & \begin{minipage}{7em}\centering independence\\ number\end{minipage} & \begin{minipage}{9em}\centering induced matching\\ number\end{minipage} \\ \midrule
		1 & $2$ &$2$ \\ 
		2 & $4$ &$4$ \\ 
		3 & $9$ &$9$ \\ \bottomrule
	\end{tabular}
	\quad
	\begin{tabular}{ ccc } 
	      \multicolumn{3}{c}{$H_{\cor, \FF_2}^{\boxtimes n}$} \\[0.5em]
		\toprule
		$n$ & \begin{minipage}{7em}\centering independence\\ number\end{minipage} & \begin{minipage}{9em}\centering induced matching\\ number\end{minipage} \\ \midrule
		1 & $2$ & $2$  \\ 
		2 & $8$ & $9$ \\ 
		3 & $24$ & $32$\\ 
	\bottomrule
	\end{tabular}
	\end{center}
	\caption{Independence number and induced matching number for small powers of the cap set hypergraph $H_{\capset}$ and corner hypergraph $H_{\cor, \FF_2}$. Interestingly, the independence number and induced matching number of powers of the cap set hypergraph are exactly equal for the powers $n = 1,2,3$. For the corner hypergraph we see that they are different already for the second and third power.}
	\label{tab:max_induced_matching}
\end{table}

\section{Lower bounds from the combinatorial degeneration method}
\label{application}
\label{lowerboundShannon}

In this section we discuss three methods to prove lower bounds on the Shannon capacity of directed $3$-uniform hypergraphs: the probabilistic method, the combinatorial degeneration method and the acyclic set method. We apply these methods to the corner problem---the problem of constructing large corner-free sets---which as a consequence gives new NOF communication protocols for the Eval problem. We begin by discussing the corner problem and its relation to NOF communication complexity.

\subsection{Corner problem, cap set problem and number on the forehead communication}
\label{cornerproblem}

\subsubsection*{Hypergraphs}

We recall the definition of directed $k$-uniform hypergraphs and basic properties of Shannon capacity on directed $k$-uniform hypergraphs. 
\begin{definition}
	\label{uniformhypegraph}
	A \emph{directed $k$-uniform hypergraph} $H$ is a pair $H = (V,E)$ where $V$ is a finite set of elements called vertices, and $E$ is a set of $k$-tuples of elements of $V$ which are called hyperedges or edges. If the set of edges $E$ is invariant under permuting the $k$ coefficients of its elements, then we may also think of $H$ as an \emph{undirected $k$-uniform hypergraph}.
\end{definition}

Let $H = (V,E)$ be a directed $k$-uniform hypergraph with $n$ vertices. The~\emph{adjacency tensor} $A$ of $H$ is defined as
\begin{align*}
\label{adj_tensor}
A_{i_1,\dots,i_k}=
\begin{cases}
1 \text{ if } i_1 = i_2 =\dots = i_k \text{ or } (i_1,\dots,i_k) \in E,\\
0 \text{ otherwise}.
\end{cases}
\end{align*}

\begin{definition}
	The \emph{strong product} of a pair of directed $k$-uniform hypergraphs $G = (V_G,E_G)$ and $H = (V_H, E_H)$ is denoted $G\boxtimes H$ and defined as follows. It is a directed $k$-uniform hypergraph with vertex set $V_G \times V_H$ and the following edge set: Any $k$ vertices $(g_1,h_1), \dots, (g_k,h_k) \in V_G \times V_H$ form an edge $((g_1,h_1), \dots, (g_k,h_k))$ if one of the following three conditions holds:
	\begin{enumerate}
		\item $g_1 =\dots = g_k$ and $(h_1,\dots,h_k) \in E_H$
		\item $(g_1,\dots ,g_k) \in E_G$ and $h_1 = \dots = h_k$
		\item $(g_1,\dots ,g_k) \in E_G$ and $(h_1,\dots, h_k) \in E_H$
	\end{enumerate} 
\end{definition}

\begin{definition} \label{independenset}
	An \emph{independent set} in a directed $k$-uniform hypergraph $H = (V,E)$ is a subset $S$ of the vertices $V$ that induces no edges, meaning for every $(e_1, \ldots, e_k) \in E$ there is an $i \in [k]$ such that $e_i \not\in S$. %
	The \emph{independence number} of $H$, denoted by $\alpha(H)$, is the maximal size of an independent set in $H$.
\end{definition}

	If $S$ and $T$ are independent sets in two directed $k$-uniform hypergraphs $G$ and $H$, respectively, then $S\times T$ is an independent set in the strong product $G \boxtimes H$. Therefore, we have the supermultiplicativity property $\alpha(G)\alpha(H) \leq \alpha(G \boxtimes H)$.
	For any directed $k$-uniform hypergraph $H$, let $H^{\boxtimes n}$ denote the $n$-fold product of $H$ with itself. %
\begin{definition}\label{Shannoncapacity} 
	The Shannon capacity of a directed $k$-uniform hypergraph $H$ %
	is defined as
	\begin{align*}
	\Theta(H)\coloneqq \lim_{n \ra \infty}(\alpha(H^{\boxtimes n}))^{1/n}.
	\end{align*}
\end{definition}
By Fekete's lemma we can write $\Theta(H) = \sup_{n }(\alpha(H^{\boxtimes n}))^{1/n}$. The following proposition can be deduced directly from the definition of Shannon capacity.
\begin{proposition}
	\label{basic_property_Shannon_1}
	Suppose $H$ is a directed $k$-uniform hypergraph with $m$ vertices and there is an independent set of size $s$ in $H^{\boxtimes n}$. Then $s^{\frac{1}{n}} \leq \Theta(H) \leq m$.
\end{proposition}

\subsubsection*{Corner problem}
Let $(G,+)$ be a finite Abelian group.
A \emph{corner} in $G \times G$ is a three-element set of the form
$\{(x,y),(x+\lambda,y),(x,y+\lambda)\}$
for some $x,y,\lambda \in G$ and $\lambda \neq 0$. The element $(x,y)$ is called the center of this corner. 
Let~$\rc{G}$ be the size of the largest subset $S\subseteq G \times G$ such that no three elements in $S$ form a corner. %
The corner problem asks to determine $\rc{G}$ given $G$.

Trivially, we have the upper bound $\rc{G} \leq |G|^2$. 
The best-known general upper bound on~$\rc{G}$ comes from~\cite{Shkredov2006OnAG,Shkredov_2006}, and reads
\begin{align*}
\rc{G} \leq \frac{|G|^2}{(\log \log |G|)^{c}} \, ,
\end{align*} 
where $0<c<\frac{1}{73}$ is an absolute constant. In the finite field setting, in \cite{Lacey_07} the following better upper bound for $\rc{G}$ with $G = \FF_{2}^{n}$ was obtained:
\begin{align*}
\rc{\FF_{2}^n} \leq \bigO\Bigl(|G|^2 \frac{\log\log\log |G|}{\log\log |G| }\Bigr) \, .
\end{align*} 

We may phrase the corner problem as a hypergraph independence problem.
We define $H_{\cor,G} = (V,E)$ to be the directed 3-uniform hypergraph with
$V = \{(g_1,g_2) : g_1,g_2 \in G\}$ and $E = \{((g_1,g_2),(g_1+\lambda,g_2),(g_1,g_2+\lambda)): g_1,g_2,\lambda \in G, \, \lambda \neq 0  \}$.
Then by construction:
\begin{lemma}
	\label{lem:corner_indep_set}
	$\rc{G^n} = \alpha(H_{\cor,G}^{\boxtimes n})$.
\end{lemma}
As a consequence,
$\rc{G^n} = \Theta(H_{\cor,G})^{n - o(n)} \, .$

\begin{example}
	Let $G$ correspond to addition in $\FF_{2}$. Then $H_{\cor, G} = (V,E)$ with
	\begin{align*}
	E = \{((0,0),(1,0),(0,1)),((0,1),(1,1),(0,0)),((1,0),(0,0),(1,1)),((1,1),(0,1),(1,0)) \}.
	\end{align*}
	Under the labeling $(0,0)=0,(0,1)=1,(1,0)=2$ and $(1,1)=3$ we will think of $H_{\cor,\FF_{2}}$ as the hypergraph $H_{\cor,\FF_{2}} = (V,E)$ with $V = (0,1,2,3)$ and $E = \{(0,2,1),(1,3,0),(2,0,3), (3,1,2) \}$. 
\end{example}  
Closely related to $\rc{G}$ is the minimum number of colors needed to color~$G \times G$ so that no corner is monochromatic, which we denote by $\cc{G}$. Then:
\begin{proposition}[\cite{CFL83,LPS18}] %
\label{prop:coloring_indep_set}
	Let $(G,+)$ be a finite Abelian group. There is a constant~$c$, such that for every $n \in \NN$,
	\label{prop:corner_and_color}
	\begin{align*}
	\frac{|G|^{2n}}{\rc{G^n}} \leq \cc{G^n} \leq c\frac{n|G|^{2n}\log |G|}{\rc{G^n}} \, .
	\end{align*}
\end{proposition}
For $G  = \FF_{2}$, the current upper bound in the literature is $\cc{\FF_{2}^{n}} \leq \bigO(n2^{n/2})$~\cite{LPS18}, which we will improve on.

\subsubsection*{Number on the forehead communication}

The corner problem is closely related to the Number On the Forehead (NOF) communication model~\cite{CFL83}.
In this model, $k$ players wish to evaluate a function $F: \mathcal{X}_1 \times \dots \times \mathcal{X}_k \ra \{0,1\}$ on a given input $x_1,\dots,x_k$. The input is distributed among the players in a way that player~$i$ sees every $x_j$ for $j \neq i$. This scenario is visualized as $x_i$ being written on the forehead of Player~$i$. The computational power of everyone is unlimited, but the number of exchanged bits has to be minimized. 
Let $\dcc{k}(F)$ be the minimum number of bits they need to communicate to compute the function $F$ in the NOF model with $k$ players. 
Many questions that have been thoroughly analyzed for the two-player case remain open in the general case of $3$ or more players, where lower bounds on communication complexity are much more difficult to prove. 
The difficulty in proving lower bounds arises from the overlap in the inputs known to different players. 

One interesting function in this context is the family of Eval functions. The function $\eval{G^n}: (G^n)^3 \ra \{0,1\}$ outputs $1$ on inputs $x_1, x_2, x_3 \in G^n$ if and only if $x_1 +x_2+ x_3 = 0$. The trivial algorithm gives that $\dcc{3}(\eval{G^n}) \leq \lceil n\log(|G|) \rceil+1$. For two players Yao \cite{Yao79} proved that $\dcc{2}(\eval{G^n}) = \Omega(n)$ (for nontrivial $G$). But, for three players it is an open problem whether $\dcc{3}(\eval{G^n}) = \Omega(n)$.

\begin{lemma}[\cite{beigel2006multiparty}]
	\label{corner_nof_eval}
	$\log(\cc{G^n}) \leq \dcc{3}(\eval{G^n}) \leq 2+\log(\cc{G^n}) \,.$  
\end{lemma}

From Lemma~\ref{corner_nof_eval} and Proposition~\ref{prop:corner_and_color} it follows that $\Theta(H_{\cor,G})<|G|^2$ would imply that $\dcc{3}(\eval{G^n}) = \Omega(n)$, and also that lower bounds on $\rc{G^n}$ give  upper bounds on $\dcc{3}(\eval{G^n})$. 
For $G = \FF_{2}$, the best-known upper bound on $\dcc{3}(\eval{\FF_{2}^n})$ is $0.5n+\bigO(\log n)$ \cite{ACFN15} which we improve on.

\subsubsection*{Three-term arithmetic progressions and the cap set problem}
\label{capsetproblem}

A \emph{three-term arithmetic progression} in $G$ is a three-element set of the form 
$\{x, x+\lambda, x + 2\lambda\}$
for some $x, \lambda \in G$ and $\lambda \neq 0$. 
Let $r_3(G)$ be the size of the largest subset $S \subseteq G$ such that no three elements in $S$ form a three-term arithmetic progression.

Following \cite[Corollary 3.24]{zhao} there is a simple relation between corner-free sets and three-term-arithmetic-progression-free sets:

\begin{lemma}\label{lem:cap-sets-corners}
	$p^n\, r_3(\FF_p^n) \leq \rc{\FF_p^n}$
\end{lemma}
\begin{proof}
	Let $S \subseteq \FF_p^n$ be a subset that is free of three-term arithmetic progressions. Define the subset $T = \{(x,y) : x - y \in S\}$. Then $T$ is a corner-free set of size $p^n |S|$. Indeed, if $(x,y), (x + \lambda, y), (x, y + \lambda)$ are elements of $T$, then $x-y, x+\lambda-y, x-y-\lambda$ are in $S$ and these elements form a three-term arithmetic progression.
\end{proof}

A three-term-arithmetic-progression-free subset of $\FF_3^n$ is also called a \emph{cap set}.
The notorious cap set problem is to determine how $r_3(\FF_3^n)$ grows when~$n$ goes to infinity. A priori we have that $2^n \leq r_3(\FF_3^n) \leq 3^n$.
Using Fourier methods and the density increment argument of Roth, the upper bound $r_3(\FF_3^n) \leq \bigO(3^n/n)$ was obtained by Meshulam \cite{MESHULAM1995168}, and improved only as late as 2012 to $\bigO(3^n/n^{1+\epsilon})$ for some positive constant $\epsilon$ by Bateman and Katz in \cite{bateman2012new}. Until recently it was not known whether $r_3(\FF_3^n)$ grows like $3^{n-\smallo(n)}$ or like $c^{n-\smallo(n)}$ for some $c<3$. Gijswijt and Ellenberg solved this question in 2017, showing that $r_3(\FF_3^n) \leq 2.756^{n + o(n)}$~\cite{EG17}.
The best lower bound is $2.2174^{n} \leq r_3(\FF_3^n)$ by Edel~\cite{Ed04}.
In particular, using Lemma~\ref{lem:cap-sets-corners}, this implies the lower bound $3^n \cdot 2.2174^n = 6.6522^n \leq \rc{\FF_3^n}$ for the corner problem. We will improve this lower bound in Theorem~\ref{cor:lower_F3}.

We may phrase the cap set problem as a hypergraph independence problem by defining the undirected $3$-uniform hypergraph $H_{\capset}$ consisting of three vertices $\{0,1,2\}$ and a single edge $e = \{0,1,2\}$. 
The independence number $\alpha(H_{\capset}^{\boxtimes n})$ equals $r_3(\FF_3^n)$, and thus the Shannon capacity of $H_{\capset}$ determines the rate of growth of $r_3(\FF_3^n)$.

\subsection{Probabilistic method} %
We start off with a simple and general method for obtaining lower bounds on the Shannon capacity. For any element $g \in G$, the set $\{(g,g+\lambda): \lambda \in G \}$ is an independent set of $H_{\cor,G}$, and therefore we have $\Theta(H_{\cor,G}) \geq |G|$, which we think of as the trivial lower bound. By using a simple probabilistic argument (which does not use much of the structure of $H_{\cor,G}$), we show the following nontrivial lower bound for $\Theta(H_{\cor,G})$.
\begin{proposition}
	\label{prop:simple_lowerbound}
	For any finite Abelian group $G$, we have $\Theta(H_{\cor,G}) \geq |G|^{3/2}$.  
\end{proposition}
\begin{proof}
	Let $|G| = m$ and $n \in \NN$. 
	Recall that the hypergraph $H_{\cor,G}^{\boxtimes n}$ has vertices given by the elements of $G^n \times G^n$ and edges given by the corners in $G^n \times G^n$. 
	Let $p =  1/{\sqrt{3(m^n-1)}}$ and choose the subset $A$ of $V(H_{\cor,G}^{\boxtimes n})$ randomly by choosing any element $(g_1, g_2) \in G^n \times G^n$ to be in the set $A$ with probability $p$. Let $H_A$ be the directed subhypergraph of $H_{\cor,G}^{\boxtimes n}$ induced by~$A$. We have $\mathbb{E}[|V(H_A)|] = m^{2n}p$. 
	Let $e$ be any edge of $H_{\cor,G}^{\boxtimes n}$. Then $e$ is of the form
	\[e = \big((g_1,g_2),(g_1+\lambda,g_2),(g_1,g_2+\lambda) \big)\]
	for some $g_1, g_2, \lambda \in G^n$ and $\lambda \neq 0$. Since $(g_1,g_2)$, $(g_1 + \lambda, g_2)$ and $(g_1, g_2 + \lambda)$ are different, and for each the probability of being in $A$ is $p$, we have that 
	 $\Pr[e \in E(H_A)] = p^3$. 
	Therefore, since $|E(H_{\cor,G}^{\boxtimes n})| = m^{2n}(m^n-1)$, we have $\mathbb{E}[|E(H_A)|] = m^{2n}(m^n-1)p^3$. 
	On the other hand, for any hypergraph $H$ we have $\alpha(H) \geq |V(H)| - |E(H)|$. 
	Therefore
	\begin{align*}
	\alpha(H_{\cor,G}^{\boxtimes n}) \geq \mathbb{E}[|V(H_A)|] - \mathbb{E}[|E(H_A)|] = \frac{2m^{2n}}{3\sqrt{3(m^{n}-1)}} \, .
	\end{align*}
	Thus find the lower bound $\Theta(H_{\cor,G}) = \lim_{n \ra \infty}\alpha(H_{\cor,G}^{\boxtimes n})^{1/n} \geq m^{3/2} \, $.     
\end{proof}
The idea in the proof of Proposition~\ref{prop:simple_lowerbound} to apply the probabilistic method to lower bound the number of remaining elements afther a ``pruning'' procedure (in this case, pruning vertices that induce edges) goes back to \cite{coppersmith1987matrix}. A similar probabilistic method construction is the driving component in the recent new upper bound on the matrix multiplication exponent~$\omega$~\cite{DBLP:conf/soda/AlmanW21}.

In terms of the corner problem, the lower bound on the Shannon capacity in Proposition~\ref{prop:simple_lowerbound} for $G = \FF_{2}$ corresponds to the upper bound $\cc{\FF_{2}^{n}} \leq \bigO(n2^{n/2})$ (via Proposition~\ref{prop:corner_and_color}). This upper bound is similar to the bound provided in~\cite[Corollary 26 in the ITCS version]{LPS18}.
\begin{remark}\label{rem:multiplayer}
	The proof of Proposition~\ref{prop:simple_lowerbound} directly extends from 2-dimensional corners to $k$-dimensional corners, which are sets of the form 
	\[
		\{(x_1, x_2, \ldots, x_k), (x_1 + \lambda, x_2,\ldots, x_k), \ldots, (x_1, x_2, \ldots, x_k + \lambda)\}.
	\]
	Just like the Eval problem on 3 players is closely related to 2-dimensional corners in $(G^n)^{\times 2}$, the Eval function on $k+1$ players is closely related to $k$-dimensional corners in $(G^n)^{\times k}$. 
	By a similar argument as the proof of Lemma~\ref{corner_nof_eval} we have that the $k+1$ player NOF complexity is upper bounded by $\dcc{k+1}(\eval{G^n}) \leq k+ \cck{k}{G^n}$, where $\cck{k}{G^n}$ is minimum number of colors that we can use to color $(G^{n})^{\times k}$ such that no $k$-dimensional corner is monochromatic. Letting $\rck{k}{G^n}$ denote the size of the largest $k$-dimensional corner free set in ${G^n}^{\times k}$, we have similar to Proposition~\ref{prop:coloring_indep_set} the relation between $\rck{k}{G^n}$ and $\cck{k}{G^n}$ given by 
	\[
		\frac{|G|^{kn}}{\rck{k}{G^n}} \leq \cck{k}{G^n} \leq \frac{nk|G|^n\log(|G|)}{\rck{k}{G^n}},
	\]
	which is proved in~\cite{LPS18}. From a similar probabilistic method argument as in the proof of Proposition~\ref{prop:simple_lowerbound}, choosing each $(x_1,\dots,x_k) \in (G^n)^{\times k}$ independently at random with probability $p=\frac{1}{[(k+1)(|G|^n-1)]^{1/k}}$, we get
	\begin{align*}
		\rck{k}{G^n} \geq \frac{k|G|^{kn}}{|G|^{n/k}(k+1)^{\frac{k+1}{k}}} \, , 
	\end{align*}
	as a consequence one has $\Theta(H_{k,\cor,G}) \geq |G|^{k-1/k}$, where $H_{k,\cor,G}$ is directed $(k+1)$-uniform hypergraph that construct for the $k$-dimensional corner. Furthermore from the lower bound of $\rck{k}{G^n}$, we have 
	\begin{align*}
		\dcc{k+1}(\eval{G^n})  \leq \frac{n}{k} \log |G|+\log n +\log\log |G| + (1+\frac{1}{k})\log(1+k)+k \, .
	\end{align*}
	If we take $k = \log n$ (for instance), then $\dcc{k+1}(\eval{G^n}) \leq \frac{n}{\log n}\log |G| +\bigO(\log n)$, that is, we obtain a sublinear upper bound for $\dcc{\log n}(\eval{G^n})$ in $n$.
\end{remark}

\subsection{Combinatorial degeneration method}\label{subsubsec:comb-degen}
We now introduce the combinatorial degeneration method for lower bounding Shannon capacity. Combinatorial degeneration is an existing concept from algebraic complexity theory introduced by Strassen in \cite[Section~6, in particular Theorem 6.1]{Strassen1987RelativeBC}%
\footnote{The precise connection to \cite{Strassen1987RelativeBC} is as follows. Strassen defines the notion of \emph{M-degeneration} on tensors. In our terminology, a tensor is an $M$-degeneration of another tensor, if the support of the first is a combinatorial degeneration of the support of the second. The terminology ``combinatorial degeneration'', which does not refer to tensors, but rather directly to their supports (hence the adjective ``combinatorial''), was introduced in~\cite[Definition~15.29]{PMM97}.}. `'
In that original setting it was used as part of the construction of fast matrix multiplication algorithms \cite[Definition~15.29 and Lemma~15.31]{PMM97}, and, in a broader setting, combinatorial degeneration was used to construct large induced matchings in \cite[Lemma~3.9]{Alon2013}, \cite[Lemma~5.1]{DBLP:conf/innovations/AlmanW18} and \cite[Theorem~4.11]{CVZ18}. 
However, we will be using it in a novel manner in order to construct independent sets instead of induced matchings.
We will subsequently apply the combinatorial degeneration method  to get new bounds for the corner problem. We expect the method to be useful in the study of other problems besides the corner problem as well. First we must define combinatorial degeneration. 
\begin{definition}[Combinatorial degeneration]\label{def:comb-degen}
	Let $I_1, \ldots, I_k$ be finite sets.
	Let $\Phi \subseteq \Psi \subseteq I_1 \times \dots \times I_k$. We say that $\Phi$ is a \emph{combinatorial degeneration} of $\Psi$, and write~$\Psi \unrhd \Phi$, if there are maps $u_i:I_i \ra \Z$ ($i \in [k]$) such that for every $x = (x_1, \ldots, x_k) \in I_1 \times \dots \times I_k$, if~$x \in \Psi \setminus \Phi$, then $\sum_{i=1}^{k}u_i(x_i)>0$, and if $x \in \Phi$, then $\sum_{i=1}^{k}u_i(x_i) = 0$.  
\end{definition}

\begin{example}
	As a quick example of a combinatorial degeneration, let
	\begin{align*}
	\Phi &= \{(0,0,0),(1,1,0),(1,0,1) \}, \\
	\Psi &= \{(0,0,0),(1,1,0),(1,0,1),(0,1,1)\}.
	\end{align*}
	Then we have a combinatorial degeneration $\Psi \unrhd \Phi$ by picking the maps $u_1(0) = u_2(0) = u_3(0) = 0$, and $u_1(1) = -1$, $u_2(1) = u_3(1) = 1$. 
\end{example}

We apply combinatorial degeneration in the following fashion to get Shannon capacity lower bounds:

\begin{theorem}[Combinatorial degeneration method]\label{thm:combinatorial_degeneration}
Let $H = (V,E)$ be a directed $k$-uniform hypergraph. Let $S \subseteq V$. Let $\Psi = E \cup\{(v,\ldots, v) : v\in V\}$ and let $\Phi = \{(v,\ldots, v) : v \in S\}$ and suppose that $\Psi \unrhd \Phi$. Then $\Theta(H) \geq \abs{S}$.
\end{theorem}
\begin{proof}
	Let $u_i$ be the maps given by the combinatorial degeneration $\Psi \unrhd \Phi$.
	Let $n$ be any multiple of $|S|$. Let $(x^{(1)}, \ldots, x^{(k)}) \in \Psi^{\otimes n}$.
	Suppose for every $i \in [k]$ that the $n$ elements in the tuple $x^{(i)} = (x^{(i)}_1, \ldots, x^{(i)}_n)$ are uniformly distributed over $S$, so that every element of~$S$ appears $n/|S|$ times in $x^{(i)}$. Then, using that $\sum_{i=1}^k u_i(s) = 0$ for every $s \in S$ and the uniformity of $x^{(i)}$, we have
	\begin{equation}
		\label{eq:is_zero}
		\sum_{i=1}^k \sum_{j=1}^n u_i(x^{(i)}_j) = \frac{n}{|S|} \sum_{s \in S} \sum_{i=1}^k u_i(s) = 0.
	\end{equation}
	For every $j \in [n]$, since $(x_j^{(1)}, \ldots, x_j^{(k)}) \in \Psi$, we have $\sum_{i=1}^k u_{i}(x_j^{(i)}) \geq 0$. Suppose that there is an index $j \in [n]$ such that $(x^{(1)}_j, \ldots, x^{(k)}_j) \not\in\Phi$. Then
	\[
	\sum_{i=1}^k u_i(x_j^{(i)}) > 0.
	\]
	As a consequence, $\sum_{j=1}^n \sum_{i=1}^k u_i(x_j^{(i)}) > 0$, which contradicts \eqref{eq:is_zero}.
	Thus the uniform strings in~$S^n$ form an independent set in $H^{\boxtimes n}$. There are
	\[
		\binom{|S| \frac{n}{|S|}}{\frac{n}{|S|}, \ldots, \frac{n}{|S|}} \geq \frac{|S|^n}{(n + 1)^{|S|}}
	\]
	such strings. The inequality follows from the fact that the largest multinomial coefficient is the central one, that is, $\binom{n}{n_1, \ldots, n_{|S|}} \leq \binom{n}{\frac{n}{|S|}, \ldots, \frac{n}{|S|}}$ and the number of possible partitions of $n$ into $|S|$ parts is at most $(n+1)^{|S|}$.
\end{proof}
Motivated by Theorem~\ref{thm:combinatorial_degeneration} we have the following definition:
\begin{definition}\label{def:beta}
	For any directed $k$-uniform hypergraph $H = (V,E)$, %
	we define $\beta(H)$ to be the size of the largest subset $S\subseteq V$ such that $\{(v,\dots,v): v\in S \}$ is a combinatorial degeneration of $E \cup \{(v,\dots,v): v\in V \}$. 
\end{definition}
Clearly, $\Theta(H) \geq \beta(H)$ by Theorem~\ref{thm:combinatorial_degeneration}.

In order to construct combinatorial degenerations we employ integer programming.
To state the integer program, we let $t$ be a variable that takes values in $\{0,1\}^{|V|}$ and let $u_1,\dots,u_k$ be variables that take values in $\Z^{|V|}$. We choose $M \in \NN$ large enough. The parameter $\beta(H)$ can be then computed by the following integer linear program:

\begin{equation}
	\label{pro:combinatorial_degeration}
	\boxed{
		\begin{array}{rrll}
			\text{max } & \sum_{i \in V}t(i) \\
			\text{subject to } & u_1(i_1)+\dots+u_k(i_k)  \!\!&\geq 1  & \forall (i_1,\dots,i_k) \in E,  \\
			& 1- t(i) \leq u_1(i)+\dots+u_k(i) \!\!&\leq M(1-t(i))  & \forall i \in V %
		\end{array}	
	}
\end{equation}

 Indeed, if $(t,u_1,\dots,u_k)$ is a feasible solution of the program~\eqref{pro:combinatorial_degeration}, then $\{(v,\dots,v): v\in S\}$ is a combinatorial degeneration of $E \cup \{(v,\dots,v): v\in V\}$ by choosing $k$ integer maps $u_1,\dots,u_k$, where $S = \{i \in V:t(i)=1 \}$. Therefore, one has $\beta(H) \geq A$ ($A$ is a maximum value of program~\eqref{pro:combinatorial_degeration}). On the other hand, for any $S \subseteq V$ such that if there is a combinatorial degeneration from $E \cup \{(v,\dots,v): v\in V\}$ to $\{(v,\dots,v):v\in S \}$ with $k$ integer maps $u_1,\dots,u_k$, by defining $t \in \{0,1\}^{|V|}$ so that $t(i) =1 $ iff $i\in S$, we have $(t,u_1,\dots,u_k)$ is a feasible solution of the program~\eqref{pro:combinatorial_degeration}. Thus,  $\beta(H) \leq A$.

As a first application of the combinatorial degeneration method (Theorem~\ref{thm:combinatorial_degeneration}), we prove the following new bound for corners over~$\FF_3^n$ by lower bounding $\beta$ (Definition~\ref{def:beta}). %

\begin{theorem}
	\label{cor:lower_F3}
	$\beta(H_{\cor,\FF_{3}}) \geq 7$ and thus $\Theta(H_{\cor,\FF_{3}}) \geq 7$. %
\end{theorem}
In other words, $7^{n}/\poly(n) \leq \rc{\FF_3^n}$. This improves on the lower bound $6.6522^n \leq \rc{\FF_3^n}$ that can be obtained from Edel's construction of cap sets \cite{Ed04} and Lemma~\ref{lem:cap-sets-corners}.
As a consequence of the new lower bound, we find the bounds $\cc{\FF_{3}^n} \leq \bigO(\poly(n)(\frac{9}{7})^{n})$ and $\dcc{3}(\eval{\FF_{3}^n}) \leq n\log(9/7) + \bigO(\log n) \leq 0.37n + \bigO(\log n)$. Previously, only the weaker bound $\dcc{3}(\eval{\FF_{3}^n}) \leq n + \bigO(\log n)$ was known \cite{LPS18}.\footnote{We note that, as far as we know, the NOF protocol for $\eval{\FF_{2}^n}$ given in~\cite{ACFN15} does not generalize to $\eval{\FF_{3}^n}$ in any direct way.}
\begin{proof}
	Let $\Psi$ be the support of the adjacency tensor of $H_{\cor,\FF_{3}}$. We label each pair $(a,b)$ for $a,b \in \{0,1,2\}$ by the integer number $3a+b$. 
	The hypergraph $H_{\cor,\FF_{3}}$ has vertex set $V = \{0,1,3,4,5,6,7,8\}$ and the set $\Psi$ is given by
	\begin{align*}
	\Psi = \{&(0,0,0),(1,1,1),(2,2,2),(3,3,3),(4,4,4),(5,5,5),(6,6,6),(7,7,7), (8,8,8),\\&(0,3,1),(0,6,2),(1,4,2),(1,7,0),(2,5,0),(2,8,1),
	(3,6,4),(3,0,5),(4,7,5),\\
	&(4,1,3), (5,8,3),(5,2,4), (6,0,7),(6,3,8),(7,1,8), (7,4,6),(8,2,6),(8,5,7) \} \, .
	\end{align*} 
	Let $S \subseteq V(H_{\cor,\FF_{3}})$ be the subset consisting of the following seven vertices:
	\begin{align*}
	S \coloneqq \{0,1,2,3,4,7,8 \} \, .
	\end{align*}
	One directly verifies that the maps $u_i : V \to \Z$ provided in the following table give a combinatorial degeneration from $\Psi$ to $\Phi_S\coloneqq \{(v,v,v) : v \in S\}$.
	
	\begin{center}\footnotesize
		\begin{tabular}{ crrr } 
			\toprule
			vertex & $u_1$ & $u_2$ &$u_3$\\
			\midrule
			$0$ & $-5$ & $1$ & $4$\\
			$1$ & $-5$ & $4$ & $1$\\
			$2$ & $-5$ & $1$ & $4$\\
			$3$ & $-5$ & $5$ & $0$\\
			$4$ & $-3$ & $2$ & $1$\\
			$5$ & $-1$ & $5$ & $5$\\
			$6$ & $-1$ & $5$ & $5$\\
			$7$ & $-3$ & $2$ & $1$\\
			$8$ & $-5$ & $5$ & $0$\\ 
			\bottomrule
		\end{tabular}
	\end{center} 
	We conclude that $\beta(H_{\cor,\FF_{3}}) \geq 7$ (Definition~\ref{def:beta}) and thus $\Theta(H_{\cor,\FF_{3}}) \geq 7$ by Theorem~\ref{thm:combinatorial_degeneration}.
\end{proof}

In the previous proof we only considered the first power of the relevant hypergraph. For the next result we will be able to get good bounds by considering higher powers.

\begin{theorem}
	\label{cor:lower_sqrt10}
	$\beta(H_{\cor,\FF_{2}}^{\boxtimes 2}) \geq 11$ and $\beta(H_{\cor,\FF_{2}}^{\boxtimes 3}) \geq 39, $ as a consequence $\Theta(H_{\cor,\FF_{2}}) \geq \sqrt[3]{39} \, $.  %
\end{theorem}
In other words, $(\sqrt[3]{39})^{n}/\poly(n) \leq \rc{\FF_2^n}$. As a consequence, we have the upper bound $\cc{\FF_{2}^n} \leq \bigO(\poly(n)(\frac{4}{\sqrt[3]{39}})^n)  \leq \bigO(\poly(n)1.18^n)$ for the corner problem and the upper bound $\dcc{3}(\eval{\FF_{2}^n}) \leq \log(\frac{4}{\sqrt[3]{39}}) n + \bigO(\log n)  \leq 0.24n + \bigO(\log n)$ for the eval problem.

\begin{proof}
	Let $H = H_{\cor,\FF_{2}}^{\boxtimes 2}$. We will show $\beta(H) \geq 11$. %
	Let $\Psi$ be the support of the adjacency tensor of $H$. Then $\Psi$ is the following set of 64 triples:
	\begin{align*}
		\!\!\!\!\!\Psi = \{&((0,0), (0,0), (0,0)),\,\,
		((0,1), (0,1), (0,1)),\,\,
		((0,2), (0,2), (0,2)),\,\,
		((0,3), (0,3), (0,3)),\,\,\\
		&((1,0), (1,0), (1,0)),\,\,
		((1,1), (1,1), (1,1)),\,\,
		((1,2), (1,2), (1,2)),\,\,
		((1,3), (1,3), (1,3)),\,\,\\
		&((2,0), (2,0), (2,0)),\,\,
		((2,1), (2,1), (2,1)),\,\,
		((2,2), (2,2), (2,2)),\,\,
		((2,3), (2,3), (2,3)),\,\,\\
		&((3,0), (3,0), (3,0)),\,\,
		((3,1), (3,1), (3,1)),\,\,
		((3,2), (3,2), (3,2)),\,\,
		((3,3), (3,3), (3,3)),\,\,\\
		&((0,0), (0,2), (0,1)),\,\,
		((0,0), (2,0), (1,0)),\,\,
		((0,0), (2,2), (1,1)),\,\,
		((0,1), (0,3), (0,0)),\,\,\\
		&((0,1), (2,1), (1,1)),\,\,
		((0,1), (2,3), (1,0)),\,\,
		((0,2), (0,0), (0,3)),\,\,
		((0,2), (2,0), (1,3)),\,\,\\
		&((0,2), (2,2), (1,2)),\,\,
		((0,3), (0,1), (0,2)),\,\,
		((0,3), (2,1), (1,2)),\,\,
		((0,3), (2,3), (1,3)),\,\,\\
		&((1,0), (1,2), (1,1)),\,\,
		((1,0), (3,0), (0,0)),\,\,
		((1,0), (3,2), (0,1)),\,\,
		((1,1), (1,3), (1,0)),\,\,\\
		&((1,1), (3,1), (0,1)),\,\,
		((1,1), (3,3), (0,0)),\,\,
		((1,2), (1,0), (1,3)),\,\,
		((1,2), (3,0), (0,3)),\,\,\\
		&((1,2), (3,2), (0,2)),\,\,
		((1,3), (1,1), (1,2)),\,\,
		((1,3), (3,1), (0,2)),\,\,
		((1,3), (3,3), (0,3)),\,\,\\
		&((2,0), (0,0), (3,0)),\,\,
		((2,0), (0,2), (3,1)),\,\,
		((2,0), (2,2), (2,1)),\,\,
		((2,1), (0,1), (3,1)),\,\,\\
		&((2,1), (0,3), (3,0)),\,\,
		((2,1), (2,3), (2,0)),\,\,
		((2,2), (0,0), (3,3)),\,\,
		((2,2), (0,2), (3,2)),\,\,\\
		&((2,2), (2,0), (2,3)),\,\,
		((2,3), (0,1), (3,2)),\,\,
		((2,3), (0,3), (3,3)),\,\,
		((2,3), (2,1), (2,2)),\,\,\\
		&((3,0), (1,0), (2,0)),\,\,
		((3,0), (1,2), (2,1)),\,\,
		((3,0), (3,2), (3,1)),\,\,
		((3,1), (1,1), (2,1)),\,\,\\
		&((3,1), (1,3), (2,0)),\,\,
		((3,1), (3,3), (3,0)),\,\,
		((3,2), (1,0), (2,3)),\,\,
		((3,2), (1,2), (2,2)),\,\,\\
		&((3,2), (3,0), (3,3)),\,\,
		((3,3), (1,1), (2,2)),\,\,
		((3,3), (1,3), (2,3)),\,\,
		((3,3), (3,1), (3,2)) \}.	
	\end{align*}
	 Let $S \subseteq V(H)$ be the subset consisting of the following eleven vertices of $H$:
	\begin{align*}
	S \coloneqq \{(0,0), (0,1), (1,0), (1,2), (1,3,) (2,0), (2,1), (2,2), (3,1), (3,2), (3,3)  \} \, .
	\end{align*}
	One directly verifies that the maps $u_i : \{0,1,2,3\}^2 \to \Z$ provided in the following table give a combinatorial degeneration from $\Psi$ to $\Phi_S\coloneqq \{(v,v,v) : v \in S\}$. %
	\begin{center}\footnotesize
		\begin{tabular}{ crrr } 
			\toprule
			vertex & $u_1$ & $u_2$ &$u_3$\\
			\midrule
			$(0,0)$ & $-10$ & $0$ & $10$\\
			$(0,1)$ & $-10$ & $0$ & $10$\\
			$(0,2)$ & $10$ & $1$ & $-1$\\
			$(0,3)$ & $10$ & $1$ & $1$\\
			$(1,0)$ & $-10$ & $0$ & $10$\\
			$(1,1)$ & $-8$ & $10$ & $10$\\
			$(1,2)$ & $7$ & $3$ & $-10$\\
			$(1,3)$ & $1$ & $5$ & $-6$\\
			$(2,0)$ & $-5$ & $1$ & $4$\\
			$(2,1)$ & $-6$ & $1$ & $5$\\
			$(2,2)$ & $9$ & $1$ & $-10$\\
			$(2,3)$ & $10$ & $3$ & $-7$\\
			$(3,0)$ & $-3$ & $1$ & $10$\\
			$(3,1)$ & $-8$ & $1$ & $7$\\
			$(3,2)$ & $8$ & $1$ & $-9$\\
			$(3,3)$ & $9$ & $-1$ & $-8$\\
			\bottomrule
		\end{tabular}
	\end{center}
We make it easier to carry out the above verification by hand by listing in the following table the following data: every element $e =(v_1, v_2, v_3) \in \Psi$, the corresponding evaluation $(u_1(v_1), u_2(v_2), u_3(v_3))$, the sum of the evaluations $\sum_i u_i(v_i)$, and whether~$e$ is in $\Phi$ or in~$\Psi\setminus \Phi$.
\begin{center}
	\scriptsize
	\begin{tabular}{clcl}
		\toprule
		$(v_1, v_2, v_3)$ & $(u_1, u_2, u_3)$ & $\sum_i u_i$ & in $\Phi$ or $\Psi\setminus \Phi$?\\
		\midrule
		((0,0), (0,0), (0,0)) & $(-10, 0, 10)$ & 0 & $\Phi$\\
		((0,1), (0,1), (0,1)) & $(-10, 0, 10)$ & 0 & $\Phi$\\
		((0,2), (0,2), (0,2)) & $(10, 1, -1)$ & 10 & $\Psi\setminus\Phi$\\
		((0,3), (0,3), (0,3)) & $(10, 1, 1)$ & 12 & $\Psi\setminus\Phi$\\
		((1,0), (1,0), (1,0)) & $(-10, 0, 10)$ & 0 & $\Phi$\\
		((1,1), (1,1), (1,1)) & $(-8, 10, 10)$ & 12 & $\Psi\setminus\Phi$\\
		((1,2), (1,2), (1,2)) & $(7, 3, -10)$ & 0 & $\Phi$\\
		((1,3), (1,3), (1,3)) & $(1, 5, -6)$ & 0 & $\Phi$\\
		((2,0), (2,0), (2,0)) & $(-5, 1, 4)$ & 0 & $\Phi$\\
		((2,1), (2,1), (2,1)) & $(-6, 1, 5)$ & 0 & $\Phi$\\
		((2,2), (2,2), (2,2)) & $(9, 1, -10)$ & 0 & $\Phi$\\
		((2,3), (2,3), (2,3)) & $(10, 3, -7)$ & 6 & $\Psi\setminus\Phi$\\
		((3,0), (3,0), (3,0)) & $(-3, 1, 10)$ & 8 & $\Psi\setminus\Phi$\\
		((3,1), (3,1), (3,1)) & $(-8, 1, 7)$ & 0 & $\Phi$\\
		((3,2), (3,2), (3,2)) & $(8, 1, -9)$ & 0 & $\Phi$\\
		((3,3), (3,3), (3,3)) & $(9, -1, -8)$ & 0 & $\Phi$\\
		((0,0), (0,2), (0,1)) & $(-10, 1, 10)$ & 1 & $\Psi\setminus\Phi$\\
		((0,0), (2,0), (1,0)) & $(-10, 1, 10)$ & 1 & $\Psi\setminus\Phi$\\
		((0,0), (2,2), (1,1)) & $(-10, 1, 10)$ & 1 & $\Psi\setminus\Phi$\\
		((0,1), (0,3), (0,0)) & $(-10, 1, 10)$ & 1 & $\Psi\setminus\Phi$\\
		((0,1), (2,1), (1,1)) & $(-10, 1, 10)$ & 1 & $\Psi\setminus\Phi$\\
		((0,1), (2,3), (1,0)) & $(-10, 3, 10)$ & 3 & $\Psi\setminus\Phi$\\
		((0,2), (0,0), (0,3)) & $(10, 0, 1)$ & 11 & $\Psi\setminus\Phi$\\
		((0,2), (2,0), (1,3)) & $(10, 1, -6)$ & 5 & $\Psi\setminus\Phi$\\
		((0,2), (2,2), (1,2)) & $(10, 1, -10)$ & 1 & $\Psi\setminus\Phi$\\
		((0,3), (0,1), (0,2)) & $(10, 0, -1)$ & 9 & $\Psi\setminus\Phi$\\
		((0,3), (2,1), (1,2)) & $(10, 1, -10)$ & 1 & $\Psi\setminus\Phi$\\
		((0,3), (2,3), (1,3)) & $(10, 3, -6)$ & 7 & $\Psi\setminus\Phi$\\
		((1,0), (1,2), (1,1)) & $(-10, 3, 10)$ & 3 & $\Psi\setminus\Phi$\\
		((1,0), (3,0), (0,0)) & $(-10, 1, 10)$ & 1 & $\Psi\setminus\Phi$\\
		((1,0), (3,2), (0,1)) & $(-10, 1, 10)$ & 1 & $\Psi\setminus\Phi$\\
		((1,1), (1,3), (1,0)) & $(-8, 5, 10)$ & 7 & $\Psi\setminus\Phi$\\
((1,1), (3,1), (0,1)) & $(-8, 1, 10)$ & 3 & $\Psi\setminus\Phi$\\
((1,1), (3,3), (0,0)) & $(-8, -1, 10)$ & 1 & $\Psi\setminus\Phi$\\
((1,2), (1,0), (1,3)) & $(7, 0, -6)$ & 1 & $\Psi\setminus\Phi$\\
((1,2), (3,0), (0,3)) & $(7, 1, 1)$ & 9 & $\Psi\setminus\Phi$\\
((1,2), (3,2), (0,2)) & $(7, 1, -1)$ & 7 & $\Psi\setminus\Phi$\\
((1,3), (1,1), (1,2)) & $(1, 10, -10)$ & 1 & $\Psi\setminus\Phi$\\
((1,3), (3,1), (0,2)) & $(1, 1, -1)$ & 1 & $\Psi\setminus\Phi$\\
((1,3), (3,3), (0,3)) & $(1, -1, 1)$ & 1 & $\Psi\setminus\Phi$\\
((2,0), (0,0), (3,0)) & $(-5, 0, 10)$ & 5 & $\Psi\setminus\Phi$\\
((2,0), (0,2), (3,1)) & $(-5, 1, 7)$ & 3 & $\Psi\setminus\Phi$\\
((2,0), (2,2), (2,1)) & $(-5, 1, 5)$ & 1 & $\Psi\setminus\Phi$\\
((2,1), (0,1), (3,1)) & $(-6, 0, 7)$ & 1 & $\Psi\setminus\Phi$\\
((2,1), (0,3), (3,0)) & $(-6, 1, 10)$ & 5 & $\Psi\setminus\Phi$\\
((2,1), (2,3), (2,0)) & $(-6, 3, 4)$ & 1 & $\Psi\setminus\Phi$\\
((2,2), (0,0), (3,3)) & $(9, 0, -8)$ & 1 & $\Psi\setminus\Phi$\\
((2,2), (0,2), (3,2)) & $(9, 1, -9)$ & 1 & $\Psi\setminus\Phi$\\
((2,2), (2,0), (2,3)) & $(9, 1, -7)$ & 3 & $\Psi\setminus\Phi$\\
((2,3), (0,1), (3,2)) & $(10, 0, -9)$ & 1 & $\Psi\setminus\Phi$\\
((2,3), (0,3), (3,3)) & $(10, 1, -8)$ & 3 & $\Psi\setminus\Phi$\\
((2,3), (2,1), (2,2)) & $(10, 1, -10)$ & 1 & $\Psi\setminus\Phi$\\
((3,0), (1,0), (2,0)) & $(-3, 0, 4)$ & 1 & $\Psi\setminus\Phi$\\
((3,0), (1,2), (2,1)) & $(-3, 3, 5)$ & 5 & $\Psi\setminus\Phi$\\
((3,0), (3,2), (3,1)) & $(-3, 1, 7)$ & 5 & $\Psi\setminus\Phi$\\
((3,1), (1,1), (2,1)) & $(-8, 10, 5)$ & 7 & $\Psi\setminus\Phi$\\
((3,1), (1,3), (2,0)) & $(-8, 5, 4)$ & 1 & $\Psi\setminus\Phi$\\
((3,1), (3,3), (3,0)) & $(-8, -1, 10)$ & 1 & $\Psi\setminus\Phi$\\
((3,2), (1,0), (2,3)) & $(8, 0, -7)$ & 1 & $\Psi\setminus\Phi$\\
((3,2), (1,2), (2,2)) & $(8, 3, -10)$ & 1 & $\Psi\setminus\Phi$\\
((3,2), (3,0), (3,3)) & $(8, 1, -8)$ & 1 & $\Psi\setminus\Phi$\\
((3,3), (1,1), (2,2)) & $(9, 10, -10)$ & 9 & $\Psi\setminus\Phi$\\
((3,3), (1,3), (2,3)) & $(9, 5, -7)$ & 7 & $\Psi\setminus\Phi$\\
((3,3), (3,1), (3,2)) & $(9, 1, -9)$ & 1 & $\Psi\setminus\Phi$\\
\bottomrule
\end{tabular}
\end{center}
\newpage
Indeed, in the above table we see that $\sum_i u_i(v_i)$ is always nonnegative, and equals 0 if and only if $(v_1, v_2, v_3) \in \Phi$. Therefore $\beta(H) \geq 11$, finishing this part.

Now we give the construction that implies $\beta(H_{\cor,\FF_{2}^n}^{\boxtimes 3}) \geq 39$. Let $S \subseteq V(H_{\cor,\FF_{2}^n}^{\boxtimes 3})$ be the subset consisting of the following thirty-nine vertices:
\begin{align*}
	S \coloneqq \{&(0,0,0), (0,0,1), (0,0,2), (0,1,0), (0,1,2), (0,1,3), (0,2,1), (0,2,3), (0,3,0), (0,3,1),\\ &(0,3,3), (1,0,2), (1,1,0), (1,1,1), (1,1,3), (1,2,1), (1,2,2), (1,2,3), (1,3,1), (1,3,2),\\ &(1,3,3), (2,0,0), (2,0,1), (2,0,3), (2,1,2), (2,1,3), (2,2,0), (2,2,2), (2,3,0), (2,3,2),\\ &(2,3,3), (3,0,0), (3,0,1), (3,0,2), (3,1,2), (3,2,0), (3,2,1), (3,2,3), (3,3,2)\}.
\end{align*}
One directly verifies that the maps $u_i : \{0,1,2,3\}^3 \to \Z$ provided in the following table give a combinatorial degeneration from $\Psi = \{(v,v,v): v \in V(H_{\cor,\FF_{2}^n}^{\boxtimes 3})\} \cup E(H_{\cor,\FF_{2}^n}^{\boxtimes 3})$ to $\Phi_S\coloneqq \{(v,v,v) : v \in S\}$.
\begin{center}
	\scriptsize
	\begin{tabular}{ clcl } 
		\toprule
		vertex & $u_1$ & $u_2$ &$u_3$\\
		\midrule
		$(0,0,0)$ & $1$ & $-4$ & $3$\\
		$(0,0,1)$ & $-15$ & $-4$ & $19$\\
		$(0,0,2)$ & $-3$ & $-17$ & $20$\\
		$(0,0,3)$ & $20$ & $20$ & $20$\\
		$(0,1,0)$ & $-5$ & $7$ & $-2$\\
		$(0,1,1)$ & $20$ & $20$ & $20$\\
		$(0,1,2)$ & $-5$ & $20$ & $-15$\\
		$(0,1,3)$ & $1$ & $0$ & $-1$\\
		$(0,2,0)$ & $20$ & $5$ & $20$\\
		$(0,2,1)$ & $5$ & $-4$ & $-1$\\
		$(0,2,2)$ & $20$ & $20$ & $7$\\
		$(0,2,3)$ & $1$ & $18$ & $-19$\\
		$(0,3,0)$ & $-4$ & $3$ & $1$\\
		$(0,3,1)$ & $20$ & $-20$ & $0$\\
		$(0,3,2)$ & $20$ & $20$ & $20$\\
		$(0,3,3)$ & $20$ & $-20$ & $0$\\
		$(1,0,0)$ & $20$ & $20$ & $20$\\
		$(1,0,1)$ & $20$ & $20$ & $20$\\
		$(1,0,2)$ & $-2$ & $-14$ & $16$\\
		$(1,0,3)$ & $9$ & $20$ & $20$\\
		$(1,1,0)$ & $-16$ & $20$ & $-4$\\
		$(1,1,1)$ & $0$ & $1$ & $-1$\\
		$(1,1,2)$ & $20$ & $20$ & $20$\\
		$(1,1,3)$ & $-4$ & $5$ & $-1$\\
		$(1,2,0)$ & $20$ & $20$ & $20$\\
		$(1,2,1)$ & $-5$ & $1$ & $4$\\
		$(1,2,2)$ & $18$ & $2$ & $-20$\\
		$(1,2,3)$ & $20$ & $0$ & $-20$\\
		$(1,3,0)$ & $20$ & $20$ & $20$\\
		$(1,3,1)$ & $-4$ & $-11$ & $15$\\
		$(1,3,2)$ & $1$ & $2$ & $-3$\\
		$(1,3,3)$ & $20$ & $-15$ & $-5$\\
		\bottomrule
	\end{tabular} \quad\quad
	\begin{tabular}{ clcl } 
	\toprule
	vertex & $u_1$ & $u_2$ &$u_3$\\
	\midrule
	$(2,0,0)$ & $5$ & $-5$ & $0$\\
	$(2,0,1)$ & $-4$ & $4$ & $0$\\
	$(2,0,2)$ & $20$ & $20$ & $20$\\
	$(2,0,3)$ & $-15$ & $5$ & $10$\\
	$(2,1,0)$ & $20$ & $20$ & $20$\\
	$(2,1,1)$ & $20$ & $20$ & $20$\\
	$(2,1,2)$ & $-9$ & $7$ & $2$\\
	$(2,1,3)$ & $1$ & $1$ & $-2$\\
	$(2,2,0)$ & $2$ & $5$ & $-7$\\
	$(2,2,1)$ & $20$ & $20$ & $20$\\
	$(2,2,2)$ & $9$ & $1$ & $-10$\\
	$(2,2,3)$ & $20$ & $20$ & $20$\\
	$(2,3,0)$ & $-4$ & $0$ & $4$\\
	$(2,3,1)$ & $20$ & $20$ & $20$\\
	$(2,3,2)$ & $4$ & $-10$ & $6$\\
	$(2,3,3)$ & $12$ & $-9$ & $-3$\\
	$(3,0,0)$ & $15$ & $-16$ & $1$\\
	$(3,0,1)$ & $-19$ & $6$ & $13$\\
	$(3,0,2)$ & $-3$ & $-17$ & $20$\\
	$(3,0,3)$ & $3$ & $20$ & $20$\\
	$(3,1,0)$ & $-1$ & $19$ & $20$\\
	$(3,1,1)$ & $20$ & $20$ & $13$\\
	$(3,1,2)$ & $-17$ & $-3$ & $20$\\
	$(3,1,3)$ & $20$ & $20$ & $20$\\
	$(3,2,0)$ & $-2$ & $4$ & $-2$\\
	$(3,2,1)$ & $8$ & $7$ & $-15$\\
	$(3,2,2)$ & $9$ & $20$ & $20$\\
	$(3,2,3)$ & $10$ & $0$ & $-10$\\
	$(3,3,0)$ & $20$ & $20$ & $20$\\
	$(3,3,1)$ & $20$ & $20$ & $20$\\
	$(3,3,2)$ & $-7$ & $-2$ & $9$\\
	$(3,3,3)$ & $20$ & $20$ & $20$\\
	\bottomrule
\end{tabular}
\end{center}
This implies the claim.
\end{proof}

We have yet to develop structural understanding of how the above combinatorial degenerations that exhibit the new capacity lower bounds arise, and leave the investigation of further generalizations and improvements to future work. 
As a partial remedy to our limited understanding, we introduce in the next section the \emph{acyclic method} as a tool to construct combinatorial degenerations. While the acyclic method does not recover the bounds of Theorem~\ref{cor:lower_sqrt10} and Theorem~\ref{cor:lower_F3}, it has the merits of being transparent and simple to apply.

\begin{remark}
	The above proof of Theorem~\ref{thm:combinatorial_degeneration} gives in fact the precise lower bound
	\begin{equation}\label{eq:precise-lower}
		\alpha(H^{\boxtimes n}) \geq \frac{|S|^n}{(n + 1)^{|S|}}.
	\end{equation}
	This lower bound is optimal up to a $\poly(n)$ factor. The following more careful analysis improves this $\poly(n)$ factor, but may safely be skipped when the reader is satisfied by the lower bound of \eqref{eq:precise-lower}.
	
	We may without loss of generality assume that $S = V$. 
	For $p \in \Z$, let $[V^{n}]^{(i)}_p \subseteq V^n$ be the subset of all elements $(x_1, \ldots, x_n) \in V^n$ such that $\sum_{j=1}^n u_i(x_j) = p$.
	For $p_1, \ldots, p_k \in \Z$, we let $[\Psi^{\otimes n}]_{p_1, \ldots, p_k} \subseteq \Psi^{\otimes n}$ denote the subset of all elements $(x^{(1)}, \ldots, x^{(k)}) \in \Psi^{\otimes n}$ such that for every $i \in [k]$ we have $\sum_{j=1}^n u_i(x_j^{(i)}) = p_i$. Thus $[\Psi^{\otimes n}]_{p_1, \ldots, p_k} = \Psi^{\otimes n} \cap ([V^n]^{(1)}_{p_1} \times \cdots \times [V^n]^{(k)}_{p_1})$.
	Then 
	\[
		\Psi^{\otimes n} = \bigsqcup_{p_1, \ldots, p_k} [\Psi^{\otimes n}]_{p_1, \ldots, p_k}
	\]
	and from the definition of a combinatorial degeneration we get
	\begin{equation}\label{eq:Phi}
		\Phi^{\otimes n} = \bigsqcup_{\substack{p_1, \ldots, p_k:\\ \sum_{i=1}^k p_i = 0}} [\Psi^{\otimes n}]_{p_1, \ldots, p_k}.
	\end{equation}
	Since $\Phi^{\otimes n}$ only contains elements of the form $(x, \ldots, x)$, we see that if $[\Psi^{\otimes n}]_{p_1, \ldots, p_k} \neq \emptyset$ and $\sum_{i=1}^k p_i = 0$, then the elements of $[\Psi^{\otimes n}]_{p_1, \ldots, p_k}$ are all the elements $(x, \ldots, x)$ going over all $x \in \cap_{i=1}^k [V^n]^{(i)}_{p_i}$.
	Thus $\alpha(H^{\otimes n}) \geq \abs{[\Psi^{\otimes n}]_{p_1, \ldots, p_k}}$ for any choice of $p_1, \ldots, p_k$ such that $\sum_{i=1}^k p_i = 0$. 
	
	One good choice of $p_1, \ldots, p_k$ is obtained as follows, and lets us recover the lower bound in~\eqref{eq:precise-lower}. For notational simplicity we are still assuming $S = V$. Let $(x_1, \ldots, x_n) \in V^n$ be any element that is uniform on $S$. For every $i \in [k]$ let $p_i = \sum_{j=1}^n u_i(x_j)$. Note that for every $i \in [k]$ the value of $p_i$ remains the same if we had picked another uniform element $(x_1, \ldots, x_n) \in V^n$.
	We claim that $\sum_{i=1}^k p_i = 0$. To prove this, let $(x^{(1)}, \ldots, x^{(k)}) \in \Psi^n$ be any element for which every $x^{(i)}$ is uniform on $S$. Then we have $p_1 + \cdots + p_k = \sum_i \sum_j u_i(x^{(i)}_j) = 0$, using that for every $s \in S$ we have $\sum_i u_i(s) = 0$.
	Finally, note that $[\Psi^{\otimes n}]_{p_1, \ldots, p_k}$ contains all elements $(x^{(1)}, \ldots, x^{(k)}) \in \Psi^{\otimes n}$ for which every $x^{(i)}$ is uniform. Therefore, with this choice we recover a bound that is at least as good as \eqref{eq:precise-lower}. 
	
	Another choice of $p_1, \ldots, p_k$ (that leads to an incomparable lower bound) is obtained as follows. Note that if $[\Psi^{\otimes n}]_{p_1, \ldots, p_k} \neq \emptyset$, then $n\min_{x \in V} u_i(x) \leq p_i \leq n\max_{x \in V} u_i(x)$. Thus the number of nonzero summands in \eqref{eq:Phi} is at most $c_{|S|} n^{k-1}$ for a constant $c_{|S|}$ that depends only on $|S|$. Therefore, there is a choice of $p_1, \ldots, p_k$ with $\sum_{i=1}^k p_i = 0$ such that
	\[
		\alpha(H^{\boxtimes n}) \geq \abs{[\Psi^{\otimes n}]_{p_1, \ldots, p_k}} \geq \frac{|\Phi^{\otimes n}|}{c_{|S|}n^{k-1}} = \frac{|S|^n}{c_{|S|} n^{k-1}}
	\]
	which improves on \eqref{eq:precise-lower} in some parameter regimes.
\end{remark}

\subsection{Acyclic set method}
\label{sec:acyclic}
The acyclic set method that we are about to introduce is modeled on the fact that
the Shannon capacity of a directed graph $G$ is at least the size of any induced acyclic subgraph of $G$ \cite{BM85}. 
We introduce the concept of an \emph{acyclic set} in a directed $k$-uniform hypergraph as an extension of the notion of an induced acyclic subgraph.

\begin{definition}
	Let $H$ %
	be a directed $k$-uniform hypergraph.
	We associate to $H$ the directed graph $G_H$ with vertices $V(G) = V(H)$ and edges $E(G) = \{(a_1, a_2) : (a_1, a_2, \ldots, a_k) \in E \text{ for some }  a_3, \dots, a_k \}$.
	For any subset $A \subseteq V$ let $H[A]$ denote the subhypergraph of $H$ induced by $A$, that is, $H[A]$ is the directed $k$-uniform hypergraph with vertices $S$ and edges $E \cap A^{\times k}$. 
	We call a subset $A \subseteq V$ an \emph{acyclic set} of $H$ if the directed graph $G_{H[A]}$ is a directed acyclic graph. %
\end{definition}

Note that, if $A$ is an independent set of $H$, then $E(H[A]) = \emptyset$ and thus $E(G_{H[A]}) = \emptyset$, and in particular $A$ is an acyclic set of $H$.
On the other hand, acyclic sets are not necessarily independent sets. However, the existence of an acyclic set does imply strong lower bounds on the Shannon capacity (via combinatorial degeneration, as we will see):

\begin{theorem} \label{thm:lowerbound_acyclic}
	Let $H$ be a directed $k$-uniform hypergraph. For any acyclic set $A$ of $H$, we have
	$\Theta(H) \geq |A|$.
\end{theorem}

Theorem~\ref{thm:lowerbound_acyclic} follows directly from the combinatorial degeneration method (Theorem~\ref{thm:combinatorial_degeneration}) and the following lemma:

\begin{lemma}
	\label{Lem:Acyclicset_degeneration}
	Let $H = (V,E)$ be a directed $k$-uniform hypergraph. 
	Let $A$ be an acyclic set of~$H$. 
	Then there is a combinatorial degeneration from $E \cup \{(v,\dots,v): v \in V\}$ to $\Phi = \{(v,\dots,v): v\in A\}$. 
\end{lemma}

\begin{proof}
We may assume that $A = V = [n]$. The proof for the case that $A \subsetneq V$ is a simple adaptation.
Recall that we construct the directed graph $G$ associated to $H$ with the same vertex set as $H$ and the edges as follows:
for every edge  $e = (a_1,a_2,\dots,a_k)$ in $H$ we add the edge $(a_1,a_2)$ to $G$. 
Since $V$ is an acyclic set we have that $G$ is a directed acyclic graph. Therefore, we have a topological ordering on the vertices of $G$. A topological ordering is a total ordering~$>$ on the vertices such that if $(u,v)$ forms an edge then $u < v$. Assume that this ordering is $1>2>\dots >n$. 
For each vertex $i \in [n]$, we define $u_1(i) = -i$, $u_2(i) = i$, $u_3(i) = \cdots = u_k(i) =0$. For every $i \in [n]$ we clearly have $u_1(i)+u_2(i)+ \dots + u_k(i) = 0$. For each edge $e = (a_1,a_2,\dots,a_k)$ in $H$ we have $u_1(a_1) + u_2(a_2) + \dots +u_k(a_k)>0$ because of the topological ordering and since we have the edge $(a_1,a_2)$ in $G$. Therefore we have a combinatorial degeneration from $E \cup \{(v,\dots,v): v\in V\}$ to $\{(v,\dots,v): v\in V\}$. For the case $A\subsetneq V$ the proof is similar except that we define $u_1(i),u_2(i),\dots,u_k(i)$ to be some large integer number for each $i \in V \setminus A$.
\end{proof}

As can be seen from the proof of Lemma~\ref{Lem:Acyclicset_degeneration}, the combinatorial degenerations that result from acyclic sets have a special form, and in particular the acyclic set method does not recover the full power of the combinatorial degeneration method. However the acyclic set method is much easier to apply than the combinatorial degeneration method. 
For example, we can use the acyclic set method to quickly see that $\Theta(H_{\cor,\FF_{2}}) \geq 3$.
	Namely, it is verified directly that the set $S = \{0,1,2\}$ of size three is an acyclic set in $H_{\cor,\FF_{2}}$, which implies the claim by Theorem \ref{thm:lowerbound_acyclic}. %

Finally, we note that %
for directed graphs ($k=2$) the combinatorial degeneration method can be used to characterize whether the Shannon capacity is full or not.
\begin{theorem}
	Let $G = (V,E)$ be a directed graph. Then $\Theta(G) =  |V|$ if and only if there is a combinatorial degeneration from $E \cup \{(v,v): v \in V\}$ to $\{(v,v): v\in V\}$. 
\end{theorem}
\begin{proof}
	The \emph{if} direction follows directly from Theorem~\ref{thm:combinatorial_degeneration}. For the \emph{only if} direction, it is shown in~\cite{BM85} that if $\Theta(G) = |V|$, then G is an acyclic graph. Then, applying Lemma~\ref{Lem:Acyclicset_degeneration} for $k=2$ proves the claim.    
\end{proof}

\section{Upper bounds from tensor methods, and their limitations}
\label{sec:ind-match}
In this section we discuss methods to obtain upper bounds on the Shannon capacity of directed $k$-uniform hypergraphs and we discuss limitations of these methods for hypergraphs like $H_{\cor,G}$. We will not be discussing all available methods, but rather some of the main ones: subrank and slice rank. The main point is to introduce the induced matching barrier and apply it to the corner problem.

We recall some standard tensor notation and definitions that we will use in the rest of the section.
For $d \in \NN$ let $[d] = \{1,\dots,d\}$. 
Let $\cP([d])$ be the set of all probability distributions on~$[d]$. 
Let $f\in \FF^{d_1} \otimes \dots \otimes \FF^{d_k}$ be a $k$-tensor over a field $\FF$. 
Let $\{e_1,\dots,e_{d_j}\}$ denote the standard basis of $\FF^{d_j}$. We may then write $f$ as
\[
	f = \sum f_{i_1,\dots,i_k}\,e_{i_1}\otimes \dots e_{i_k},
\]
where the sum goes over $i \in [d_1] \times \cdots \times [d_k]$. 
In this way $f$ %
corresponds to a $k$-way array $f \in \FF^{d_1 \times \dots \times d_k}$. For $f\in \FF^{d_1} \otimes \dots \otimes \FF^{d_k}$ and $f' \in  \FF^{d'_1} \otimes \dots \otimes \FF^{d'_k}$, we define the tensor product as $(f \otimes f')_{(i_1, j_1), \dots, (i_k, j_k)} = f_{i_1, \dots, i_k} \cdot f'_{j_1, \dots, j'_k}$.
We define the \emph{support} of $f$ as the set
\[
	\supp(f)\coloneqq \{(i_1,\dots,i_k): f_{i_1,\dots,i_k} \neq 0 \} \subseteq [d_1] \times \dots \times [d_k].
\]
For $r \in \NN$, we call $\left \langle  r \right \rangle\coloneqq \sum_{i=1}^{r}e_{i}^{\otimes k}$ the \emph{unit tensor} of size~$r$.

\subsection{Tensor methods: subrank, slice rank (and more)}
\label{tensorrank_inducedmatching}
We focus on two tensor methods here: subrank and slice rank. We begin by defining subrank, for which we need the notion of
restriction of tensors~\cite{Strassen1987RelativeBC}.
	We say that the tensor $f \in  \FF^{d_1} \otimes \dots \otimes \FF^{d_k}$ \emph{restricts to} $f'\in  \FF^{d'_1} \otimes \dots \otimes \FF^{d'_k}$, and write $f' \leq f$ if there exist linear maps $A^{(i)}: \FF^{d_i} \ra \FF^{d'_i}$ such that $f' = (A^{(1)} \otimes \cdots \otimes A^{(k)})\cdot f$. Written in the standard basis, this corresponds to having for all $i_1 \in [d'_1], \dots, i_k \in [d'_k]$ that
	\begin{align*}
	f'_{i_1,\dots,i_k} = \sum_{j_1 \in[d_1],\dots,j_k \in [d_k]}A^{(1)}_{i_1,j_1} \dots A^{(k)}_{i_k,j_k}f_{j_1,\dots,j_k}.
	\end{align*}

\begin{example}
	Here we see restriction in action in a small example.
	For the tensors
	\begin{align*}
	f  &= e_0 \otimes e_0 \otimes e_0 + e_1 \otimes e_1 \otimes e_1,\quad 
	f' = e_0\otimes(e_0 \otimes e_0 + e_1 \otimes e_1), 
	\end{align*}
	we have $f' \leq f$ by letting $A^{(1)}: e_{0} \mapsto e_0, e_{1} \mapsto e_0$ and letting $A^{(2)}$ and $A^{(3)}$ both be the identity map.
\end{example}
Let $\langle n \rangle = \sum_{i \in [n]} e_i \otimes \cdots \otimes e_i$ be the \emph{unit tensor} of rank $n$.
	Strassen \cite{Strassen1987RelativeBC} defined the \emph{subrank} of $f$ as
	\begin{align*}
	\subrank(f) \coloneqq \max \{r\in \NN: \left \langle  r \right \rangle \leq f\}.
	\end{align*}
Similarly, one may define the ``opposite'' of the subrank as $\trank(f) \coloneqq \min \{r \in \NN: f \leq \left \langle  r \right \rangle\}$, which is called the rank and 
which coincides with the usual notion of tensor rank in terms of a rank-one decomposition. 
For $k = 2$, the subrank and rank of $f$ are the usual matrix rank: $\subrank(f) = \trank(f) = \rank(f)$. When $k \geq 3$, however, there are $f$ for which $\subrank(f) < \trank(f)$. In fact, the tensor rank can be larger than the dimensions $d_1, \dots, d_k$, whereas the subrank cannot exceed $\min_i d_i$.

Applications require us to understand the rate of growth of the subrank as we take tensor product powers of a fixed tensor. Strassen~\cite{Strassen1987RelativeBC} defined the \emph{asymptotic subrank} of $f \in \FF^{d_1} \otimes \dots \otimes \FF^{d_k}$ as 
\begin{align*}
\asympsubrank(f): = \lim_{n \ra \infty} \subrank(f^{\otimes n})^{1/n} \, .
\end{align*}
Since the subrank is super-multiplicative, we can, by Fekete's lemma, replace the limit by a supremum. %

The second tool we focus on is slice rank.
Slice rank was introduced by Tao \cite{Tao16} and developed further in~\cite{TS16} and \cite{Matrix_capset} as a variation on tensor rank to study cap sets and approaches to fast matrix multiplication algorithms. A tensor in $\FF^{d_1}\otimes \dots \otimes \FF^{d_k}$ has \emph{slice rank one} if it has the form $u\otimes v$ for $u\in \FF^{d_i}$ and $v \in \bigotimes_{j \neq i} \FF^{d_j}$ for some $i \in [k]$. The \emph{slice rank} of $f$, denoted by $\slicerank(f)$, is the smallest number $r$ such that $f$ can be written as sum of $r$ slice rank one tensors. Since slice rank is not sub-multiplicative and not super-multiplicative, the limit $\lim_{n \ra \infty}\slicerank(f^{\otimes n})^{1/n}$ might not exist~\cite{CVZ18}. We define
\begin{align*}
\asyslicerank(f) = \limsup\limits_{n \ra \infty} \slicerank(f^{\otimes n})^{1/n} \, .
\end{align*}
Since slice rank is monotone under the restriction order and normalized on $\langle n\rangle$~\cite{Tao16}, it follows that $\subrank(f) \leq \slicerank(f)$ and $\asympsubrank(f) \leq \asyslicerank(f)$.

\subsection{Induced matchings and tightness}
Now we discuss the notion of induced matchings, and we will discuss Strassen's theorem that gives a construction of large induced matchings under a tightness condition.

Let $H = (V,E)$ be a directed $k$-uniform hypergraph with adjacency tensor $A$. 
Let $\Phi_H$ be the support of $A$. 
A subset $D \subseteq \Phi_H$ is called a \emph{matching} if any two distinct elements $a,b \in D$ differ in all $k$ coordinates, that is, $a_i \neq b_i$ for all $i \in [k]$.
We call a matching $D \subseteq \Phi_H$ an \emph{induced matching} if $D = \Phi_H \cap (D_1 \times \dots \times D_k)$, where
$D_i = \{a_i: a \in D\}$ is the projection of $D$ onto the $i$-th coordinate. 
We denote by $\subrank_{\IM}(\Phi_H)$ the maximum size of an induced matching $D \subseteq \Phi_H$. 

For two directed $k$-uniform hypergraphs $G = (V_G,E_G)$ and $H = (V_H,E_H)$, let $\Phi_G$ and $\Phi_H$ be the support of the adjacency tensors of $G$ and $H$, respectively. 
We define the product $\Phi_G \times \Phi_H \subseteq (V_G \times V_H)\times \dots \times (V_G \times V_H)$ by $\Phi_G \times \Phi_H \coloneqq \{((a_1,b_1),\dots,(a_k,b_k)): a\in \Phi_G, b\in \Phi_H \}$. The \emph{asymptotic induced matching number} of $H$ is defined as $\asympsubrank_{\IM}(\Phi_H): = \lim_{n\ra \infty}\subrank_{\IM}(\Phi_H^{\times n})^{1/n} = \sup_{n}\subrank_{\IM}(\Phi_{H}^{\times n})^{1/n}$.

The induced matching number should be thought of as the combinatorial version of the subrank, as follows.
Let $\Phi_H$ be the support of the adjacency tensor $A_H$ of a directed $k$-uniform hypergraph~$H$. 
Then the induced matching number $\subrank_{\IM}(\Phi_H)$ is the largest number $n$ such that~$\left \langle  n \right \rangle$ can be obtained from $A_H$ using a restriction that consists of matrices that have at most one nonzero entry in each row and in each column. Therefore, $\subrank_{\IM}(\Phi_{H}) \leq \subrank(A_H)$.

\begin{lemma}
	\label{matching_bound}
	Let $H$ be a directed $k$-uniform hypergraph and $A_H$ its adjacency tensor with support $\Phi_H = \supp(A_H)$. Then
	\begin{align*}
	\Theta(H) \leq \asympsubrank_{\IM}(\Phi_H) \leq \asympsubrank(A_H).
	\end{align*}
\end{lemma}  
\begin{proof}
	We begin with the first inequality.
	Let $S$ be an independent set of $H^{\boxtimes n}$. We have $\Phi_{H}^{\times n} = \supp(A_{H}^{\otimes n})$. Thus $\Phi_{H}^{\times n} \cap (S \times S \cdots \times S) = \{(a,\dots,a): a \in S\}$. This means that $|S| \leq \subrank_{\IM}(\Phi_H^{\times n})$. We conclude $\Theta(H) \leq \asympsubrank_{\IM}(\Phi_{H})$. The second inequality follows from the already established inequality $\subrank_{\IM}(\Phi_{H}) \leq \subrank(A_H)$.
\end{proof}

Next, we discuss \emph{tight sets}, a notion introduced by Strassen \cite{Str91}. 

\begin{definition}[\cite{Str91}, see also \cite{CVZ18}] \label{Tighttensor}
	Let $I_1,\dots,I_k$ be finite sets. We call any subset $\Phi \subseteq I_1 \times \dots \times I_k$ \emph{tight} if there are injective maps $u_i: I_i \ra \Z$ for every $i \in [k]$ such that:
	\begin{align*}
	u_1(a_1) + \cdots + u_k(a_k) = 0 \text{ for every } (a_1,\dots,a_k) \in \Phi.
	\end{align*}
\end{definition}

When $\Phi_H$ is tight, the asymptotic induced matching number is essentially known, and can be described as a simple optimization. To explain the precise formula we recall some definitions.

For any finite set $X$, let $\mathcal{P}(X)$ be the set of all distributions on $X$. For any probability distribution $P \in \mathcal{P}(X)$ the \emph{Shannon entropy} of $P$ is defined as $H(P)\coloneqq-\sum_{x \in X}P(x)\log_2 P(x)$ with $0\log_2 0 = 0$. Given finite sets $I_1,\dots,I_k$ and a probability distribution $P\in \mathcal{P}(I_1 \times \dots \times I_k)$ on the product set $I_1\times \dots \times I_k$ we denote the \emph{marginal distribution} of $P$ on $I_i$ by $P_i$, that is, $P_i(a)=\sum_{x:x_i=a}P(x)$ for any $a\in I_i$. 

\begin{theorem}[\cite{Str91}] \label{thm:tight}
	Let $H$ be a directed $3$-uniform hypergraph. If $\Phi_H$ is tight, then 
	\begin{align*}
	\asympsubrank_{\IM}(\Phi_{H}) = \max_{P \in \mathcal{P}(\Phi_H)} \min_{i\in[3]} 2^{H(P_i)}.
	\end{align*}
\end{theorem}
In particular, Theorem~\ref{thm:tight} implies that, for any directed 3-uniform hypergraph $H = (V,E)$ if there is a distribution $P$ on~$\Phi_H$ such that every marginal distribution $P_i$ is uniform on $V$, then $\Phi_H$ has asymptotically maximal induced matchings.

Note that Theorem~\ref{thm:tight} only applies to directed $k$-uniform hypergraphs for $k=3$. For the higher-order case $k>3$ an extension of the lower bound of Theorem~\ref{thm:tight} was proven in~\cite[Theorem~1.2.4]{DBLP:journals/cc/ChristandlVZ19}.

\subsection{The corner hypergraph is tight}
We will now apply Theorem~\ref{thm:tight} to the corner problem. First we see how the tightness property is satisfied by the corner problem by a simple construction.
\begin{theorem}\label{prop:corner-tight}
	For any finite Abelian group $(G,+)$, let $\Phi_{H_{\cor, G}}$ be the support of the adjacency tensor of $H_{\cor,G}$. Then the set $\Phi_{H_{\cor, G}}$ is tight. 
\end{theorem}
\begin{proof}
	Let $m = |G|$ and $\phi$ be a bijection between $G$ and $\{0,1,\dots, m-1\}$. We define
	\begin{align*}
	u_1((g_1,g_2)) &= \phi(g_1) + m \phi(g_2) \\
	u_2((g_1,g_2)) &= m^2 \phi(g_1+g_2) - m \phi(g_2) \\
	u_3((g_1,g_2)) &= -m^2 \phi(g_1+g_2) - \phi(g_1) \ .
	\end{align*}
	It is easy to check that the maps $u_1, u_2, u_3$ are injective and that for every triple of pairs $(g_1,g_2), (g_1+ \lambda, g_2), (g_1, g_2+\lambda)$, it holds that
	\begin{align*}
	u_1((g_1,g_2)) + u_2((g_1+ \lambda, g_2)) + u_3((g_1,  g_2+\lambda)) = 0 \, .
	\end{align*}
	This proves the claim.
\end{proof}

As a consequence of Theorem~\ref{prop:corner-tight} and Theorem~\ref{thm:tight}, we find almost directly that the asymptotic induced matching number of the corner hypergraph is maximal:

\begin{corollary}\label{cor:cor-tight}
	For any group $G$, $\asympsubrank_{\IM}(H_{\cor, G}) = |G|^2$.	
\end{corollary}
\begin{proof}
	We know that $\Phi_{H_{\cor, G}}$ is tight by Theorem~\ref{prop:corner-tight}, and so we may apply Theorem~\ref{thm:tight}. We take $P \in \mathcal{P}(\Phi_{H_{\cor, G}})$ to be the uniform probability distribution. It then suffices to observe that every marginal distribution $P_i$ is also uniform to obtain the claim.
\end{proof}

In particular, Corollary~\ref{cor:cor-tight} implies that no better upper bound on $\Theta(H_{\cor, G})$ can be obtained via methods that also upper bound the asymptotic induced matching number $\asympsubrank_{\IM}(H_{\cor, G})$. Such methods include the slice rank, the analytic rank, the geometric rank and the G-stable rank.

\bibliography{larger-corner-free-sets-itcs-arxiv}

\begin{thebibliography}{10}

\bibitem{ACFN15}
Anil Ada, Arkadev Chattopadhyay, Omar Fawzi, and Phuong Nguyen.
\newblock The {NOF} multiparty communication complexity of composed functions.
\newblock {\em Computational Complexity}, 24(3):645--694, 2015.
\newblock \href {https://doi.org/10.1007/s00037-013-0078-4}
  {\path{doi:10.1007/s00037-013-0078-4}}.

\bibitem{DBLP:conf/innovations/AlmanW18}
Josh Alman and Virginia~Vassilevska Williams.
\newblock Further limitations of the known approaches for matrix
  multiplication.
\newblock In {\em 9th Innovations in Theoretical Computer Science Conference
  (ITCS 2018)}, pages 25:1--25:15, 2018.
\newblock \href {https://doi.org/10.4230/LIPIcs.ITCS.2018.25}
  {\path{doi:10.4230/LIPIcs.ITCS.2018.25}}.

\bibitem{DBLP:conf/soda/AlmanW21}
Josh Alman and Virginia~Vassilevska Williams.
\newblock A refined laser method and faster matrix multiplication.
\newblock In {\em Proceedings of the 2021 {ACM-SIAM} Symposium on Discrete
  Algorithms (SODA 2021)}, pages 522--539, 2021.
\newblock \href {https://doi.org/10.1137/1.9781611976465.32}
  {\path{doi:10.1137/1.9781611976465.32}}.

\bibitem{Alon2013}
Noga Alon, Amir Shpilka, and Christopher Umans.
\newblock On sunflowers and matrix multiplication.
\newblock {\em computational complexity}, 22(2):219--243, 2013.
\newblock \href {https://doi.org/10.1007/s00037-013-0060-1}
  {\path{doi:10.1007/s00037-013-0060-1}}.

\bibitem{alon2020algorithmic}
Noga Alon and Adi Shraibman.
\newblock Algorithmic number on the forehead protocols yielding dense
  {Ruzsa-Szemer\'{e}di} graphs and hypergraphs, 2020.
\newblock \href {http://arxiv.org/abs/2001.00387} {\path{arXiv:2001.00387}}.

\bibitem{BGKL04}
L\'{a}szl\'{o} Babai, Anna G\'{a}l, Peter~G. Kimmel, and Satyanarayana~V.
  Lokam.
\newblock Communication complexity of simultaneous messages.
\newblock {\em SIAM J. Comput.}, 33(1):137--166, 2004.
\newblock \href {https://doi.org/10.1137/S0097539700375944}
  {\path{doi:10.1137/S0097539700375944}}.

\bibitem{bateman2012new}
Michael Bateman and Nets~Hawk Katz.
\newblock New bounds on cap sets.
\newblock {\em J. Amer. Math. Soc.}, 25(2):585--613, 2012.
\newblock \href {https://doi.org/10.1090/S0894-0347-2011-00725-X}
  {\path{doi:10.1090/S0894-0347-2011-00725-X}}.

\bibitem{beame2007separating}
Paul Beame, Matei David, Toniann Pitassi, and Philipp Woelfel.
\newblock Separating deterministic from nondeterministic {NOF} multiparty
  communication complexity.
\newblock In {\em International Colloquium on Automata, Languages, and
  Programming (ICALP~2007)}, pages 134--145, 2007.
\newblock \href {https://doi.org/10.1007/978-3-540-73420-8_14}
  {\path{doi:10.1007/978-3-540-73420-8_14}}.

\bibitem{beigel2006multiparty}
Richard Beigel, William Gasarch, and James Glenn.
\newblock The multiparty communication complexity of {Exact-{T}}: Improved
  bounds and new problems.
\newblock In {\em International Symposium on Mathematical Foundations of
  Computer Science (MFCS 2006)}, pages 146--156, 2006.
\newblock \href {https://doi.org/10.1007/11821069_13}
  {\path{doi:10.1007/11821069_13}}.

\bibitem{Bei94}
Richard Beigel and Jun Tarui.
\newblock On {ACC}.
\newblock {\em Computational Complexity}, 4(4):350--366, 1994.
\newblock \href {https://doi.org/10.1007/BF01263423}
  {\path{doi:10.1007/BF01263423}}.

\bibitem{BM85}
E.~Bidamon and H.~Meyniel.
\newblock On the {S}hannon capacity of a directed graph.
\newblock {\em European J. Combin.}, 6(4):289--290, 1985.
\newblock \href {https://doi.org/10.1016/S0195-6698(85)80042-1}
  {\path{doi:10.1016/S0195-6698(85)80042-1}}.

\bibitem{Matrix_capset}
Jonah Blasiak, Thomas Church, Henry Cohn, Joshua~A. Grochow, Eric Naslund,
  William~F. Sawin, and Chris Umans.
\newblock On cap sets and the group-theoretic approach to matrix
  multiplication.
\newblock {\em Discrete Anal.}, 2017.
\newblock \href {http://arxiv.org/abs/1605.06702} {\path{arXiv:1605.06702}},
  \href {https://doi.org/10.19086/da.1245} {\path{doi:10.19086/da.1245}}.

\bibitem{Briet2019SubspacesOT}
Jop Bri\"et.
\newblock Subspaces of tensors with high analytic rank, 2019.
\newblock \href {http://arxiv.org/abs/1908.04169} {\path{arXiv:1908.04169}}.

\bibitem{PMM97}
Peter B\"{u}rgisser, Michael Clausen, and M.~Amin Shokrollahi.
\newblock {\em Algebraic complexity theory}, volume 315 of {\em Grundlehren der
  Mathematischen Wissenschaften}.
\newblock Springer-Verlag, Berlin, 1997.
\newblock \href {https://doi.org/10.1007/978-3-662-03338-8}
  {\path{doi:10.1007/978-3-662-03338-8}}.

\bibitem{CFL83}
Ashok~K. Chandra, Merrick~L. Furst, and Richard~J. Lipton.
\newblock Multi-party protocols.
\newblock In {\em Proceedings of the 15th Annual ACM Symposium on Theory of
  Computing (STOC~1983)}, pages 94--99, 1983.
\newblock \href {https://doi.org/10.1145/800061.808737}
  {\path{doi:10.1145/800061.808737}}.

\bibitem{CS14}
Arkadev Chattopadhyay and Michael~E. Saks.
\newblock The power of super-logarithmic number of players.
\newblock In {\em Approximation, Randomization, and Combinatorial Optimization.
  Algorithms and Techniques (APPROX/RANDOM 2014)}, volume~28, pages 596--603,
  2014.
\newblock \href {https://doi.org/10.4230/LIPIcs.APPROX-RANDOM.2014.596}
  {\path{doi:10.4230/LIPIcs.APPROX-RANDOM.2014.596}}.

\bibitem{christandl2021communication}
Matthias Christandl, Omar Fawzi, Hoang Ta, and Jeroen Zuiddam.
\newblock Communication complexity, corner-free sets and the symmetric subrank
  of tensors, 2021.
\newblock \href {http://arxiv.org/abs/2104.01130v1}
  {\path{arXiv:2104.01130v1}}.

\bibitem{CVZ18}
Matthias Christandl, P{\'{e}}ter Vrana, and Jeroen Zuiddam.
\newblock Universal points in the asymptotic spectrum of tensors.
\newblock In {\em Proceedings of the 50th Annual {ACM} {SIGACT} Symposium on
  Theory of Computing (STOC~2018)}, pages 289--296, 2018.
\newblock \href {http://arxiv.org/abs/1709.07851} {\path{arXiv:1709.07851}},
  \href {https://doi.org/10.1145/3188745.3188766}
  {\path{doi:10.1145/3188745.3188766}}.

\bibitem{DBLP:journals/cc/ChristandlVZ19}
Matthias Christandl, P{\'{e}}ter Vrana, and Jeroen Zuiddam.
\newblock Asymptotic tensor rank of graph tensors: beyond matrix
  multiplication.
\newblock {\em Comput. Complex.}, 28(1):57--111, 2019.
\newblock \href {http://arxiv.org/abs/1609.07476} {\path{arXiv:1609.07476}},
  \href {https://doi.org/10.1007/s00037-018-0172-8}
  {\path{doi:10.1007/s00037-018-0172-8}}.

\bibitem{1530730}
H.~{Cohn}, R.~{Kleinberg}, B.~{Szegedy}, and C.~{Umans}.
\newblock Group-theoretic algorithms for matrix multiplication.
\newblock In {\em Proceedings of the 46th Annual IEEE Symposium on Foundations
  of Computer Science (FOCS~2005)}, pages 379--388, 2005.
\newblock \href {https://doi.org/10.1109/SFCS.2005.39}
  {\path{doi:10.1109/SFCS.2005.39}}.

\bibitem{coppersmith1987matrix}
Don Coppersmith and Shmuel Winograd.
\newblock Matrix multiplication via arithmetic progressions.
\newblock In {\em Proceedings of the nineteenth annual ACM symposium on Theory
  of computing}, pages 1--6. ACM, 1987.

\bibitem{croot2017progression}
Ernie Croot, Vsevolod~F. Lev, and P{\'e}ter~P{\'a}l Pach.
\newblock Progression-free sets in $\mathbb{Z}_4^n$ are exponentially small.
\newblock {\em Annals of Mathematics}, pages 331--337, 2017.
\newblock \href {https://doi.org/10.4007/annals.2017.185.1.7}
  {\path{doi:10.4007/annals.2017.185.1.7}}.

\bibitem{derksen2020gstable}
Harm Derksen.
\newblock The {G}-stable rank for tensors, 2020.
\newblock \href {http://arxiv.org/abs/2002.08435} {\path{arXiv:2002.08435}}.

\bibitem{Ed04}
Yves Edel.
\newblock Extensions of generalized product caps.
\newblock {\em Designs, Codes and Cryptography}, 31(1):5--14, 2004.
\newblock \href {https://doi.org/10.1023/A:1027365901231}
  {\path{doi:10.1023/A:1027365901231}}.

\bibitem{EG17}
Jordan~S. Ellenberg and Dion Gijswijt.
\newblock On large subsets of {$\mathbb{F}^n_q$} with no three-term arithmetic
  progression.
\newblock {\em Ann. of Math.}, 185(1):339--343, 2017.
\newblock \href {https://doi.org/10.4007/annals.2017.185.1.8}
  {\path{doi:10.4007/annals.2017.185.1.8}}.

\bibitem{Gargano1993}
L.~Gargano, J.~K{\"o}rner, and U.~Vaccaro.
\newblock {Sperner} capacities.
\newblock {\em Graphs and Combinatorics}, 9(1):31--46, 1993.
\newblock \href {https://doi.org/10.1007/BF01195325}
  {\path{doi:10.1007/BF01195325}}.

\bibitem{GARGANO90}
L.~Gargano, J.~Körner, and U.~Vaccaro.
\newblock Qualitative independence and sperner problems for directed graphs.
\newblock {\em Journal of Combinatorial Theory, Series A}, 61(2):173--192,
  1992.
\newblock \href {https://doi.org/10.1016/0097-3165(92)90016-N}
  {\path{doi:10.1016/0097-3165(92)90016-N}}.

\bibitem{gowers2011linear}
W.~T. Gowers and J.~Wolf.
\newblock Linear forms and higher-degree uniformity for functions on
  {$\mathbb{F}^n_p$}.
\newblock {\em Geom. Funct. Anal.}, 21(1):36--69, 2011.
\newblock \href {https://doi.org/10.1007/s00039-010-0106-3}
  {\path{doi:10.1007/s00039-010-0106-3}}.

\bibitem{Gro94}
Vince Grolmusz.
\newblock The {BNS} lower bound for multi-party protocols in nearly optimal.
\newblock {\em Inf. Comput.}, 112(1):51--54, 1994.
\newblock \href {https://doi.org/10.1006/inco.1994.1051}
  {\path{doi:10.1006/inco.1994.1051}}.

\bibitem{HW79}
Willem~H. Haemers.
\newblock On some problems of lov{\'{a}}sz concerning the shannon capacity of a
  graph.
\newblock {\em {IEEE} Trans. Inf. Theory}, 25(2):231--232, 1979.
\newblock \href {https://doi.org/10.1109/TIT.1979.1056027}
  {\path{doi:10.1109/TIT.1979.1056027}}.

\bibitem{Sawin}
Robert Kleinberg, William Sawin, and David Speyer.
\newblock The growth rate of tri-colored sum-free sets.
\newblock {\em Discrete Anal.}, 2018.
\newblock \href {http://arxiv.org/abs/1607.00047} {\path{arXiv:1607.00047}},
  \href {https://doi.org/10.19086/da.3734} {\path{doi:10.19086/da.3734}}.

\bibitem{kopparty_et_al}
Swastik Kopparty, Guy Moshkovitz, and Jeroen Zuiddam.
\newblock {Geometric Rank of Tensors and Subrank of Matrix Multiplication}.
\newblock In {\em Proceedings of the 35th Computational Complexity Conference
  (CCC 2020)}, pages 35:1--35:21, 2020.
\newblock \href {http://arxiv.org/abs/2002.09472} {\path{arXiv:2002.09472}},
  \href {https://doi.org/10.4230/LIPIcs.CCC.2020.35}
  {\path{doi:10.4230/LIPIcs.CCC.2020.35}}.

\bibitem{Lacey_07}
Michael Lacey and William McClain.
\newblock On an argument of {Shkredov} on two-dimensional corners.
\newblock {\em Online Journal of Analytic Combinatorics}, 2007.
\newblock URL:
  \url{https://hosted.math.rochester.edu/ojac/vol2/Lacey_McClain_2007.pdf},
  \href {http://arxiv.org/abs/math/0510491} {\path{arXiv:math/0510491}}.

\bibitem{LPS18}
Nati Linial, Toni Pitassi, and Adi Shraibman.
\newblock On the communication complexity of high-dimensional permutations.
\newblock {\em In 10th Innovations in Theoretical Computer Science Conference
  (ITCS 2019)}, 124, page 54:1–54:20, 2019.
\newblock \href {http://arxiv.org/abs/1706.02207} {\path{arXiv:1706.02207}},
  \href {https://doi.org/10.4230/LIPIcs.ITCS.2019.54}
  {\path{doi:10.4230/LIPIcs.ITCS.2019.54}}.

\bibitem{DBLP:conf/coco/LinialS21}
Nati Linial and Adi Shraibman.
\newblock An improved protocol for the exactly-n problem.
\newblock In Valentine Kabanets, editor, {\em 36th Computational Complexity
  Conference ({CCC} 2021)}, volume 200, pages 2:1--2:8, 2021.
\newblock \href {https://doi.org/10.4230/LIPIcs.CCC.2021.2}
  {\path{doi:10.4230/LIPIcs.CCC.2021.2}}.

\bibitem{linial2021larger}
Nati Linial and Adi Shraibman.
\newblock Larger corner-free sets from better {NOF} {E}xactly-$n$ protocols,
  2021.
\newblock \href {http://arxiv.org/abs/2102.00421} {\path{arXiv:2102.00421}}.

\bibitem{Lo79}
L{\'{a}}szl{\'{o}} Lov{\'{a}}sz.
\newblock On the shannon capacity of a graph.
\newblock {\em {IEEE} Trans. Inf. Theory}, 25(1):1--7, 1979.
\newblock \href {https://doi.org/10.1109/TIT.1979.1055985}
  {\path{doi:10.1109/TIT.1979.1055985}}.

\bibitem{Lovett}
Shachar Lovett.
\newblock The analytic rank of tensors and its applications.
\newblock {\em Discrete Anal.}, 2019.
\newblock \href {http://arxiv.org/abs/1806.09179} {\path{arXiv:1806.09179}}.

\bibitem{MESHULAM1995168}
Roy Meshulam.
\newblock On subsets of finite abelian groups with no 3-term arithmetic
  progressions.
\newblock {\em Journal of Combinatorial Theory, Series A}, 71(1):168 -- 172,
  1995.
\newblock \href {https://doi.org/10.1016/0097-3165(95)90024-1}
  {\path{doi:10.1016/0097-3165(95)90024-1}}.

\bibitem{Naslund}
Eric Naslund.
\newblock The partition rank of a tensor and $k$-right corners in
  $\mathbb{F}_{q}^{n}$.
\newblock {\em Journal of Combinatorial Theory, Series A}, 174:105190, 2020.
\newblock \href {https://doi.org/10.1016/j.jcta.2019.105190}
  {\path{doi:10.1016/j.jcta.2019.105190}}.

\bibitem{naslund2017upper}
Eric Naslund and Will Sawin.
\newblock Upper bounds for sunflower-free sets.
\newblock {\em Forum of Mathematics, Sigma}, 5, 2017.
\newblock \href {https://doi.org/10.1017/fms.2017.12}
  {\path{doi:10.1017/fms.2017.12}}.

\bibitem{sha56}
Claude~E. Shannon.
\newblock The zero error capacity of a noisy channel.
\newblock {\em {IRE} Trans. Inf. Theory}, 2(3):8--19, 1956.
\newblock \href {https://doi.org/10.1109/TIT.1956.1056798}
  {\path{doi:10.1109/TIT.1956.1056798}}.

\bibitem{Shkredov2006OnAG}
I.~D. Shkredov.
\newblock On a generalization of {S}zemer\'{e}di's theorem.
\newblock {\em Proc. London Math. Soc.}, 93(3):723--760, 2006.
\newblock \href {https://doi.org/10.1017/S0024611506015991}
  {\path{doi:10.1017/S0024611506015991}}.

\bibitem{Shkredov_2006}
I.~D. Shkredov.
\newblock On a problem of {Gowers}.
\newblock {\em Izvestiya: Mathematics}, 70(2):385--425, 2006.
\newblock \href {https://doi.org/10.1070/im2006v070n02abeh002316}
  {\path{doi:10.1070/im2006v070n02abeh002316}}.

\bibitem{SHRAIBMAN201844}
Adi Shraibman.
\newblock A note on multiparty communication complexity and the
  {Hales–Jewett} theorem.
\newblock {\em Information Processing Letters}, 139:44--48, 2018.
\newblock \href {https://doi.org/10.1016/j.ipl.2018.07.002}
  {\path{doi:10.1016/j.ipl.2018.07.002}}.

\bibitem{Strassen1987RelativeBC}
Volker Strassen.
\newblock Relative bilinear complexity and matrix multiplication.
\newblock {\em J. Reine Angew. Math.}, 375/376:406--443, 1987.
\newblock \href {https://doi.org/10.1515/crll.1987.375-376.406}
  {\path{doi:10.1515/crll.1987.375-376.406}}.

\bibitem{Str91}
Volker Strassen.
\newblock Degeneration and complexity of bilinear maps: some asymptotic
  spectra.
\newblock {\em J. Reine Angew. Math}, 413:127–180, 1991.
\newblock \href {https://doi.org/10.1515/crll.1991.413.127}
  {\path{doi:10.1515/crll.1991.413.127}}.

\bibitem{Tao16}
Terence Tao.
\newblock A symmetric formulation of the {Croot-Lev-Pach-Ellenberg-Gijswijt}
  capset bound.
\newblock {\em Tao's blog post}, 2016.
\newblock URL: \url{https://terrytao.wordpress.com}.

\bibitem{TS16}
Terence Tao and Will Sawin.
\newblock Notes on the “slice rank” of tensors.
\newblock {\em Tao's blog post}, 2016.
\newblock URL:
  \url{https://terrytao.wordpress.com/2016/08/24/notes-on-the-slice-rank-of-tensors/}.

\bibitem{DBLP:journals/sigact/Viola19}
Emanuele Viola.
\newblock Guest column: Non-abelian combinatorics and communication complexity.
\newblock {\em {SIGACT} News}, 50(3):52--74, 2019.
\newblock \href {https://doi.org/10.1145/3364626.3364637}
  {\path{doi:10.1145/3364626.3364637}}.

\bibitem{Yao79}
Andrew Chi-Chih Yao.
\newblock Some complexity questions related to distributive computing.
\newblock In {\em Proceedings of the Eleventh Annual ACM Symposium on Theory of
  Computing (STOC 1979)}, pages 209--213, 1979.
\newblock \href {https://doi.org/10.1145/800135.804414}
  {\path{doi:10.1145/800135.804414}}.

\bibitem{zhao}
Yufei Zhao.
\newblock Graph theory and additive combinatorics.
\newblock Lecture Notes, 2019.
\newblock URL: \url{https://yufeizhao.com/gtac/}.

\end{thebibliography}

\end{document}